\documentclass[twocolumn, journal]{IEEEtran}
\usepackage{amsmath, amsthm, amssymb, amsfonts, mathtools, bm, dsfont}
\usepackage{graphicx}
\usepackage{array}
\usepackage{psfrag}
\usepackage{calc}
\usepackage{float}
\usepackage{stfloats}
\usepackage{multirow}
\usepackage{subcaption}
\usepackage{enumitem}
\usepackage{microtype}
\usepackage{hyperref}
\usepackage{orcidlink}
\usepackage{pgf,tikz}
\usepackage{pgfplots}
\pgfplotsset{compat=1.17}
\usepackage{pgfplots}
\pgfplotsset{compat=1.18}
\usepackage{tikz}
\usepackage[ruled,vlined]{algorithm2e}
\SetAlFnt{\small}
\SetAlCapFnt{\small}
\SetAlCapNameFnt{\small}
\SetAlgoLined
\DontPrintSemicolon

\hypersetup{
    colorlinks=true,
    linkcolor=blue,
    citecolor=blue,
    urlcolor=blue,
}

\newtheorem{theorem}{Theorem}
\newtheorem{lemma}{Lemma}
\newtheorem{proposition}{Proposition}
\theoremstyle{definition}

\newtheorem{corollary}{Corollary}

\usetikzlibrary{svg.path}
\definecolor{orcidlogocol}{HTML}{A6CE39}
\tikzset{
  orcidlogo/.pic={
    \fill[orcidlogocol] svg{M256,128c0,70.7-57.3,128-128,128C57.3,256,0,198.7,0,128C0,57.3,57.3,0,128,0C198.7,0,256,57.3,256,128z};
    \fill[white] svg{M86.3,186.2H70.9V79.1h15.4v48.4V186.2z}
                 svg{M108.9,79.1h41.6c39.6,0,57,28.3,57,53.6c0,27.5-21.5,53.6-56.8,53.6h-41.8V79.1z M124.3,172.4h24.5c34.9,0,42.9-26.5,42.9-39.7c0-21.5-13.7-39.7-43.7-39.7h-23.7V172.4z}
                 svg{M88.7,56.8c0,5.5-4.5,10.1-10.1,10.1c-5.6,0-10.1-4.6-10.1-10.1c0-5.6,4.5-10.1,10.1-10.1C84.2,46.7,88.7,51.3,88.7,56.8z};
  }
}

\usepackage{titlesec}
\titlespacing*{\section}{0pt}{1ex plus 0.2ex minus 0.2ex}{0.8ex plus 0.2ex minus 0.2ex}
\titlespacing*{\subsection}{0pt}{0.8ex plus 0.2ex minus 0.2ex}{0.6ex plus 0.2ex minus 0.2ex}
\titlespacing*{\subsubsection}{0pt}{0.6ex plus 0.2ex minus 0.2ex}{0.4ex plus 0.2ex minus 0.2ex}

\title{Secure High-Resolution ISAC via Multi-Layer Intelligent Metasurfaces: A Layered Optimization Framework }

\author{%
Amirhossein Taherpour~\orcidlink{0000-0003-4647-102X},%
\thanks{Amirhossein Taherpour is with the Department of Electrical Engineering, Columbia University, New York, NY, USA (e-mail: at3532@columbia.edu)} 
Abbas Taherpour~\orcidlink{0000-0003-0706-5774},~\IEEEmembership{Senior Member,~IEEE},%
\thanks{Abbas Taherpour is with the Department of Electrical Engineering, Imam Khomeini International University, Qazvin, Iran (e-mail: taherpour@ikiu.ac.ir).}
and Tamer Khattab~\orcidlink{0000-0003-2347-9555},~\IEEEmembership{Senior Member,~IEEE}%
\thanks{and Tamer Khattab is with the Department of Electrical Engineering, Qatar University, Doha, Qatar (e-mail: tkhattab@ieee.org).}
}
\date{}

\begin{document}

\maketitle
\thispagestyle{empty}

\begin{abstract}
Integrated sensing and communication (ISAC) has emerged as a pivotal technology for next-generation wireless networks, enabling simultaneous data transmission and environmental sensing. However, existing ISAC systems face fundamental limitations in achieving high-resolution sensing while maintaining robust communication security and spectral efficiency. This paper introduces a transformative approach leveraging stacked intelligent metasurfaces (SIM) to overcome these challenges. We propose a multi-functional SIM-assisted system that jointly optimizes communication secrecy and sensing accuracy through a novel layered optimization framework. Our solution employs a multi-objective optimization formulation that balances secrecy rate maximization with sensing error minimization under practical hardware constraints. The proposed layered block coordinate descent algorithm efficiently coordinates sensing configuration, secure beamforming, communication metasurface optimization, and resource allocation while ensuring robustness to channel uncertainties. Extensive simulations demonstrate significant performance gains over conventional approaches, achieving 32-61\% improvement in sensing accuracy and 15-35\% enhancement in secrecy rates while maintaining computational efficiency. This work establishes a new paradigm for secure and high-precision multi-functional wireless systems.
\end{abstract}

\begin{IEEEkeywords}
Integrated sensing and communication, stacked intelligent metasurfaces, physical layer security, multi-objective optimization, robust beamforming, Cramér-Rao bound, artificial noise, resource allocation.
\end{IEEEkeywords}
\section{Introduction}

The convergence of communication and sensing functionalities represents one of the most significant paradigm shifts in beyond-5G and 6G wireless networks \cite{liu2022survey,luo2025survey}. Integrated sensing and communication (ISAC) has garnered substantial research interest due to its potential to enable diverse applications including autonomous vehicles \cite{yan2025v2x}, smart cities \cite{wei2024iotj}, and industrial IoT \cite{alqahwachi2025dynamic}. The fundamental premise of ISAC lies in the joint utilization of spectral and hardware resources for dual purposes, thereby improving spectral efficiency and reducing infrastructure costs \cite{wen2024iscc,hassan2025review}.

Recent advances in ISAC architectures have demonstrated promising performance, yet several critical challenges remain unaddressed. Conventional ISAC systems \cite{chen2023concurrent,kim2024robust} often face inherent trade-offs between sensing resolution and communication throughput, particularly in multi-user scenarios. The work in \cite{wang2025fractional} highlighted the limitations of time-sharing approaches, while \cite{gomez2023eusipco} demonstrated the vulnerability of ISAC systems to security threats. Physical layer security has emerged as a crucial consideration, with \cite{liu2025secure,zhao2025rsma} showing that eavesdroppers can exploit sensing signals to compromise communication confidentiality.

The integration of reconfigurable intelligent surfaces (RIS) has opened new avenues for enhancing ISAC performance \cite{du2022ris,liu2024joint}. RIS-assisted ISAC systems \cite{haider2023review,magbool2025survey} have demonstrated improved coverage and beamforming gains. However, single-layer RIS architectures face fundamental limitations in wave manipulation capabilities \cite{bjornson2022spm}. Recent breakthroughs in stacked intelligent metasurfaces (SIM) \cite{liu2024sim,elhaj2024sim} have revealed unprecedented opportunities for advanced wave-based signal processing. Unlike conventional RIS, multi-layer SIM can perform volumetric analog computing, enabling sophisticated spatial filtering and beamforming \cite{li2024jointsim}.

Security considerations in ISAC systems have gained increasing attention, with \cite{hoang2025an} proposing artificial-noise-based approaches and \cite{zhang2025low} developing optimization frameworks for secrecy rate maximization. However, these methods often neglect the interplay between sensing accuracy and communication security. The work in \cite{liu2024fluid} identified fundamental performance trade-offs, while \cite{liu2022moop} proposed resource allocation strategies for balancing these competing objectives.

Channel uncertainty poses another significant challenge for practical ISAC deployment. Robust optimization techniques have been explored in \cite{kim2024robust} and \cite{zhang2023icc}, but existing approaches typically assume simplified uncertainty models. The recent work in \cite{gomez2023eusipco} introduced probabilistic constraints for outage performance, while \cite{di2020practical} addressed the impact of finite-resolution phase shifters on system performance.

Multi-functional metasurface architectures represent a promising direction for overcoming current limitations. The pioneering work in \cite{li2025conformal} demonstrated dynamic beamforming capabilities, and \cite{zhang2023icc} achieved simultaneous multi-target sensing and channel estimation in RIS-assisted systems. However, these approaches lack comprehensive optimization frameworks that jointly address communication security, sensing accuracy, and resource efficiency. The survey in \cite{magbool2025survey} identified the need for integrated optimization methodologies as a critical research gap.

Recent developments in optimization theory have provided new tools for addressing ISAC challenges. Block coordinate descent and distributed methods \cite{sakurama2017distributed} have shown promise for complex large-scale problems, while multi-objective and successive convex approximation techniques \cite{liu2022moop,huang2019energy,chen2024waveform} enable efficient handling of intricate constraints. Despite these advances, several fundamental questions remain unanswered. How can we achieve simultaneous high-precision sensing and secure communication in practical scenarios with hardware constraints? What are the fundamental performance limits of SIM-assisted ISAC systems? How can we efficiently optimize multi-layer metasurface configurations while ensuring robustness against channel uncertainties?

This paper addresses these critical questions through the following contributions:
\begin{itemize}
\item We introduce a novel multi-functional SIM-assisted ISAC architecture that enables joint communication and sensing with enhanced security and accuracy.
\item We develop a comprehensive multi-objective optimization framework that simultaneously maximizes secrecy rates and minimizes sensing errors under practical constraints.
\item We propose an efficient layered block coordinate descent algorithm that coordinates sensing configuration, secure beamforming, communication metasurface optimization, and resource allocation.
\item We establish fundamental performance bounds and provide rigorous convergence guarantees for the proposed optimization framework.
\item We demonstrate through extensive simulations that our approach achieves significant performance improvements over state-of-the-art methods while maintaining computational efficiency.
\end{itemize}

The remainder of this paper is organized as follows. Section~\ref{sec:system_model} presents the system model and problem formulation. Section~\ref{sec:solution} details the proposed optimization framework. Section~\ref{sec:analysis} provides theoretical analysis. Section~\ref{sec:simulations} presents simulation results, and Section~\ref{sec:conclude} concludes the paper.

\section{System Model and Problem Formulation}
\label{sec:system_model}

We consider a multi-functional downlink system in beyond-5G networks that jointly performs communication and sensing tasks. The system comprises a base station (BS) equipped with \(M\) transmit antennas and \(R_{\mathrm{RF}}\) radio frequency (RF) chains (\(R_{\mathrm{RF}} \le M\)), a SIM with \(L\) programmable layers where each layer \(l\) contains \(N_l\) metamaterial elements, and \(K\) single-antenna legitimate users in the presence of a set \(\mathcal{E}\) of potential eavesdroppers. The total number of SIM elements is \(N_{\mathrm{tot}} = \sum_{l=1}^L N_l\). The SIM is positioned between the BS and users as a transmissive surface and is assumed to be fully controlled by the BS.

The key advantage of the multi-layer SIM architecture over conventional single-layer RIS lies in its wave-based beamforming capability. While a single-layer RIS can only apply phase shifts, the multi-layer structure enables volumetric analog processing that can implement sophisticated spatial filtering and focusing, effectively performing passive beamforming with higher gain and precision.

\subsection{Protocol and Resource Allocation}
\label{sec:protocol}

The system employs a time-division duplex protocol with three operational phases: Channel Estimation (CE), Sensing (Phase I), and Communication (Phase II), with respective time fractions $\tau_{\mathrm{ce}}$, $\tau_{\mathrm{sense}}$, and $\tau_{\mathrm{comm}}$. A switching period \(t_{\mathrm{switch}}\) is required to reconfigure the SIM between phases, implying \(\tau_{\mathrm{ce}} + \tau_{\mathrm{sense}} + \tau_{\mathrm{comm}} + 2t_{\mathrm{switch}}/T_{\mathrm{slot}} \le 1\), where \(T_{\mathrm{slot}}\) is the total slot duration. For practical SIMs, \(t_{\mathrm{switch}}\) is typically in the sub-ms range, but its impact becomes non-negligible at mmWave frequencies.

The total energy constraint per time slot is 
\begin{equation}
\begin{aligned}
& T_{\mathrm{slot}} \Big( \tau_{\mathrm{sense}} P_{\mathrm{sense}}
  + \tau_{\mathrm{comm}} \big( \sum_{k=1}^K \|\bm{w}_k\|^2 + \mathrm{Tr}(\bm{R}_{\mathrm{AN}}) \big) \Big) \\
&\le \mathcal{E}_{\mathrm{max}},
\end{aligned}
\end{equation}
where $P_{\mathrm{sense}}$ is the sensing power (an optimization variable), and the communication power is $\sum_{k=1}^K \|\bm{w}_k\|^2 + \mathrm{Tr}(\bm{R}_{\mathrm{AN}})$. Additionally, we impose instantaneous power constraints $P_{\mathrm{sense}} \le P_{\mathrm{BS}}^{\max}$ and $\sum_{k=1}^K \|\bm{w}_k\|^2 + \mathrm{Tr}(\bm{R}_{\mathrm{AN}}) \le P_{\mathrm{BS}}^{\max}$ to respect the BS power amplifier limitations.

The SIM configuration is reconfigured between phases, with \(\bm{\Phi}_l^{\mathrm{ce}}, \bm{\Phi}_l^{\mathrm{sense}}, \bm{\Phi}_l^{\mathrm{comm}} \in \mathbb{C}^{N_l \times N_l}\) denoting the diagonal phase-shift matrices for layer \(l\) during channel estimation, sensing, and communication phases, respectively. The BS employs hybrid beamforming with \(\bm{W} = \bm{W}_{\mathrm{RF}}\bm{W}_{\mathrm{BB}}\), where \(\bm{W}_{\mathrm{RF}} \in \mathbb{C}^{M \times R_{\mathrm{RF}}}\) is the analog beamforming matrix and \(\bm{W}_{\mathrm{BB}} \in \mathbb{C}^{R_{\mathrm{RF}} \times K}\) is the digital baseband precoder.

\subsection{Stacked Intelligent Metasurface Model}

The multi-layer SIM architecture enables wave-domain computing through volumetric analog processing. The composite electromagnetic transformation through all \(L\) SIM layers is defined as:
\begin{equation}\label{eq:SIM_cascade}
\begin{aligned}
&\bm{G}_{\mathrm{SIM}}(\{\bm{\Phi}_l\}) = \bm{\Phi}_L \bm{H}^{(L)} \bm{\Phi}_{L-1} \bm{H}^{(L-1)} \cdots\\[4pt]
&\quad \bm{\Phi}_1 \bm{H}^{(1)} \in \mathbb{C}^{N_L \times M},
\end{aligned}
\end{equation}
where \(\bm{H}^{(l)} \in \mathbb{C}^{N_l \times N_{l-1}}\) (with \(N_0 = M\)) represents the effective coupling matrix between layers \(l-1\) and \(l\), capturing both the free-space propagation and mutual coupling between adjacent layers.

The key advantage of the multi-layer structure is its ability to implement sophisticated beamforming patterns. For a given desired beamforming matrix \(\bm{F}_{\mathrm{desired}} \in \mathbb{C}^{N_L \times M}\), the SIM can approximate:
\begin{equation}
    \bm{G}_{\mathrm{SIM}}(\{\bm{\Phi}_l\}) \approx \bm{F}_{\mathrm{desired}},
\end{equation}
with the multi-layer structure providing higher beamforming gain and spatial resolution compared to single-layer RIS due to its increased degrees of freedom and wave processing capability.

To model practical hardware, each metasurface element has a programmable phase shift with finite resolution:
\begin{equation}
    [\bm{\Phi}_l]_{n,n} = e^{j\theta_l^{(n)}}, \quad \theta_l^{(n)} \in \mathcal{Q}_l,
\end{equation}
where \(\mathcal{Q}_l = \{ 2\pi q / 2^{b_l} \mid q=0,\ldots,2^{b_l}-1 \}\) is the discrete phase set with \(b_l\)-bit resolution. Practical hardware limitations, such as insertion loss from the metasurface elements, are assumed to be incorporated into the effective coupling matrices $\bm{H}^{(l)}$.

\subsection{Channel Models and Channel State Information (CSI) Acquisition}

For clarity we summarize the main channel notations used throughout this subsection.  \(\bm{H}^{(1)}\) denotes the BS-to-SIM near-field channel, \(\bm{H}^{(l)}\) the inter-layer SIM coupling matrix for \(l\ge 2\), \(\bm{h}_{\mathrm{SIM}\to k}\) the SIM-to-user \(k\) far-field channel, \(\bm{h}_{\mathrm{SIM}\to e}\) the SIM-to-eavesdropper \(e\) far-field channel, and \(\bm{h}_{\mathrm{eff},k}\) the effective end-to-end channel to user \(k\).

The BS-to-SIM link is modeled with near-field (spherical) wavefronts and, because of the fixed geometry of the deployment, is assumed known perfectly:
\begin{equation}
    \bm{H}^{(1)} = \bm{A}_{\mathrm{SIM}\to\mathrm{BS}} \,\bm{\Gamma}_{\mathrm{BS-SIM}}\, \bm{A}_{\mathrm{BS}\to\mathrm{SIM}}^T \in \mathbb{C}^{N_1 \times M},
\end{equation}
where the near-field steering vectors are written elementwise as \([\bm{a}_{\mathrm{BS}}(\bm{r}_n)]_m = \exp\!\big(-j\frac{2\pi}{\lambda}\|\bm{r}_n - \bm{r}_{\mathrm{BS},m}\|\big)\) and the propagation matrix captures spherical-wave propagation
\begin{equation}
    [\bm{\Gamma}_{\mathrm{BS-SIM}}]_{n,m} = \frac{\lambda}{4\pi r_{n,m}} \exp\!\Big(-j\frac{2\pi}{\lambda}r_{n,m}\Big),
\end{equation}
with \(r_{n,m}=\|\bm{r}_{\mathrm{SIM},n}-\bm{r}_{\mathrm{BS},m}\|\).

During the channel estimation phase of duration \(\tau_{\mathrm{ce}}\) the BS transmits orthogonal pilots while the SIM cycles through a predetermined set of \(Q\) configurations \(\{\bm{\Phi}_l^{\mathrm{ce},q}\}_{q=1}^Q\). The users collect the resulting measurements and estimate the effective end-to-end channels \(\bm{h}_{\mathrm{eff},k}\). Because the training budget is limited (\(Q<N_{\mathrm{tot}}\)), channel estimates are imperfect and are modeled with bounded errors:
\begin{equation}
    \hat{\bm{h}}_{\mathrm{eff},k} = \bm{h}_{\mathrm{eff},k} + \Delta\bm{h}_k,\qquad \|\Delta\bm{h}_k\| \le \epsilon_k^{\max},
\end{equation}
where the uncertainty sets \(\{\Delta\bm{h}_k\}\) are used for robust precoder and SIM configuration design in the communication phase.

The SIM-to-user and SIM-to-eavesdropper links are treated under far-field propagation with path loss and imperfect CSI:
\begin{equation}
    \bm{h}_{\mathrm{SIM}\to k}^H = \sqrt{\frac{1}{\rho_k}}\,\bm{a}_{\mathrm{SIM}}^H(\theta_k,r_k), \,\,
    \bm{h}_{\mathrm{SIM}\to e}^H = \sqrt{\frac{1}{\rho_e}}\,\bm{a}_{\mathrm{SIM}}^H(\theta_e,r_e),
\end{equation}
where \(\rho_k=(d_k/d_0)^{n}\) and \(\rho_e=(d_e/d_0)^{n}\) denote the path losses (reference distance \(d_0\) and path-loss exponent \(n\)), and \(\bm{a}_{\mathrm{SIM}}(\cdot)\) are the corresponding array steering vectors.

For inter-layer coupling (\(l\ge 2\)) we adopt a dominant line-of-sight (LoS) model to retain optimization tractability:
\begin{equation}
    \bm{H}^{(l)} = \bm{H}_{\mathrm{LoS}}^{(l)},\qquad l\ge 2.
\end{equation}
When evaluating performance via simulation we relax this assumption and use a Rician model to capture scattered components:
\begin{equation}
    \bm{H}^{(l)} = \sqrt{\kappa^{(l)}}\,\bm{H}_{\mathrm{LoS}}^{(l)} + \sqrt{1-\kappa^{(l)}}\,\bm{H}_{\mathrm{NLoS}}^{(l)},
\end{equation}
where \(\kappa^{(l)}\) is the Rician \(K\)-factor and \(\bm{H}_{\mathrm{NLoS}}^{(l)}\) contains the non-line-of-sight (NLoS) scattering contributions; this hybrid approach preserves analytical tractability while testing robustness to practical channel perturbations.

Eavesdropper uncertainty is handled probabilistically: eavesdroppers are assumed to have uncertain locations leading to random channel errors, and secrecy is enforced through an outage constraint of the form
\begin{equation}
    \Pr\big( R_k^{\mathrm{sec}} < R_k^{\min} \big) \le \epsilon_{\mathrm{out}}, \qquad \forall k,
\end{equation}
where \(\epsilon_{\mathrm{out}}\) is the maximum allowable secrecy-outage probability and uncertainties across different eavesdroppers are taken to be independent.  Together, these modeling choices (near-field deterministic BS–SIM link, bounded-user CSI errors, Rician inter-layer channels for simulation, and probabilistic eavesdropper errors) provide a balance between realistic performance assessment and optimization tractability.

\subsection{Sensing Phase Formulation}

During Phase I (sensing), the SIM processes echoes from \(T\) point targets to estimate their positions and velocities. The received sensing signal over time window \(t \in [0, T_{\mathrm{sense}}]\) is:
\begin{equation}
\begin{aligned}
&\bm{Y}_{\mathrm{sense}} = \sum_{i=1}^{T} \beta_i \,\bm{a}_{\mathrm{SIM}}(\theta_i, r_i)
    \,\bm{a}_{\mathrm{BS}}^H(\theta_i, r_i)\,\bm{X}_{\mathrm{sense}}e^{j2\pi f_{d,i}t/T_{\mathrm{sense}}} \\
&+ \bm{N}_{\mathrm{sense}},
\end{aligned}
\end{equation}
where \(\bm{N}_{\mathrm{sense}} \sim \mathcal{CN}(\bm{0}, \sigma_n^2 \bm{I})\) represents thermal noise, \(f_{d,i}\) is the Doppler frequency of target \(i\), and we assume orthogonal probing waveforms with \(\bm{X}_{\mathrm{sense}}\bm{X}_{\mathrm{sense}}^H = P_{\mathrm{sense}}\bm{I}\). Here, the sensing duration is $T_{\mathrm{sense}} = \tau_{\mathrm{sense}} T_{\mathrm{slot}}$, where $\tau_{\mathrm{sense}}$ is the sensing time fraction defined in Section \ref{sec:protocol}. The complex reflection coefficient for target \(i\) is:
\begin{equation}
    \beta_i = \sqrt{\sigma_{\mathrm{RCS},i} / \rho_i} \cdot e^{j\phi_i},
\end{equation}
where \(\sigma_{\mathrm{RCS},i}\) is the radar cross-section (RCS) and \(\rho_i\) is the path loss.

The SIM estimates target parameters \(\bm{\xi}_i = [\theta_i, r_i, v_i, \sigma_{\mathrm{RCS},i}]^T\) using maximum likelihood estimation, where \(v_i\) is the radial velocity. The Fisher Information Matrix (FIM) for parameter estimation is:
\begin{equation}
    \bm{J}(\bm{\xi}_i) = \frac{2|\beta_i|^2}{\sigma_n^2} \Re\left\{ \left( \frac{\partial \bm{\mu}(\bm{\xi}_i)}{\partial \bm{\xi}_i} \right)^H \left( \frac{\partial \bm{\mu}(\bm{\xi}_i)}{\partial \bm{\xi}_i} \right) \right\},
\end{equation}
where \(\bm{\mu}(\bm{\xi}_i) = \beta_i \bm{a}_{\mathrm{SIM}}(\theta_i, r_i) \bm{a}_{\mathrm{BS}}^H(\theta_i, r_i) \bm{X}_{\mathrm{sense}} e^{j2\pi f_{d,i}t/T_{\mathrm{sense}}}\). We compute the full FIM but note that, in our operating regime, cross-coupling among \(\theta_i, r_i, v_i\) is small; this motivates the Schur-complement-based lower bound used for optimization.

For optimization tractability, we employ a lower-bound approximation of the FIM using the Schur complement. The total normalized Cramér-Rao Bound (CRB) is:
\begin{equation}
    \mathrm{CRB}_{\mathrm{total}} = \sum_{i=1}^{T} \mathrm{Tr}\left( \bm{J}_{\mathrm{LB}}^{-1}(\bm{\xi}_i) \right),
\end{equation}
where \(\bm{J}_{\mathrm{LB}}(\bm{\xi}_i)\) is the lower-bound FIM obtained via Schur complement.

\subsection{Communication Phase Formulation}

The transmitted signal during Phase II (communication) incorporates artificial noise (AN) for security:
\begin{equation}
    \bm{x} = \bm{W}\bm{s} + \bm{z}_{\mathrm{AN}},
\end{equation}
where the data vector is \(\bm{s} \sim \mathcal{CN}(\bm{0},\bm{I}_K)\) and \(\bm{z}_{\mathrm{AN}} \sim \mathcal{CN}(\bm{0},\bm{R}_{\mathrm{AN}})\) is the AN vector.

The effective end-to-end communication channels are:
\begin{align}
    \bm{h}_{\mathrm{eff},k}^H &= \bm{h}_{\mathrm{SIM}\to k}^H \bm{G}_{\mathrm{SIM}}(\{\bm{\Phi}_l^{\mathrm{comm}}\}), \\
    \bm{h}_{\mathrm{eff},e}^H &= \bm{h}_{\mathrm{SIM}\to e}^H \bm{G}_{\mathrm{SIM}}(\{\bm{\Phi}_l^{\mathrm{comm}}\}).
\end{align}

The received signal at user \(k\) is:
\begin{equation}
    y_k = \bm{h}_{\mathrm{eff},k}^H \bm{x} + n_k,
\end{equation}
where \(n_k \sim \mathcal{CN}(0,\sigma_k^2)\). The signal-to-interference-plus-noise ratio (SINR) at user \(k\) is:
\begin{equation}
    \mathrm{SINR}_k = \frac{|\bm{h}_{\mathrm{eff},k}^H \bm{w}_k|^2}{\sum_{j\ne k} |\bm{h}_{\mathrm{eff},k}^H \bm{w}_j|^2 + \bm{h}_{\mathrm{eff},k}^H \bm{R}_{\mathrm{AN}} \bm{h}_{\mathrm{eff},k} + \sigma_k^2}.
\end{equation}

The SINR at eavesdropper \(e\) for user \(k\)'s message is:
\begin{equation}
    \mathrm{SINR}_e^{(k)} = \frac{|\bm{h}_{\mathrm{eff},e}^H \bm{w}_k|^2}{\sum_{j\ne k} |\bm{h}_{\mathrm{eff},e}^H \bm{w}_j|^2 + \bm{h}_{\mathrm{eff},e}^H \bm{R}_{\mathrm{AN}} \bm{h}_{\mathrm{eff},e} + \sigma_e^2}.
\end{equation}

The worst-case achievable secrecy rate for user \(k\) under probabilistic outage constraint is given by:
\begin{equation}
    R_k^{\mathrm{sec}} = \left[ \log_2(1+\mathrm{SINR}_k) - \max_{e\in\mathcal{E}} \log_2\left(1 + \mathrm{SINR}_e^{(k)}\right) \right]^+,
\end{equation}
with the outage constraint \(\mathbb{P}(R_k^{\mathrm{sec}} < R_k^{\min}) \leq \epsilon_{\mathrm{out}}\). Using the Bernstein-type inequality, this probabilistic constraint can be transformed into a tractable deterministic form.

\subsection{Multi-Objective Optimization Problem}

The goal is to jointly optimize SIM-assisted volumetric beamforming to enhance sensing accuracy while ensuring secure communication under limited energy and time resources. We formulate a multi-objective optimization problem that captures the fundamental trade-off between communication and sensing performance. The weighted sum approach provides the Pareto-optimal frontier.

To ensure scale compatibility between the secrecy rate (bits/s/Hz) and CRB (variance), we normalize both quantities using benchmark values:

\begin{subequations}\label{eq:optimization_problem}
\begin{align}
\lefteqn{%
  \underset{\substack{\bm{W},\, \bm{R}_{\mathrm{AN}},\, \{\bm{\Phi}_l^{\mathrm{comm}}\}, \{\bm{\Phi}_l^{\mathrm{sense}}\}, \{\bm{\Phi}_l^{\mathrm{ce}}\},\\
  \tau_{\mathrm{sense}},\, \tau_{\mathrm{ce}}, \tau_{\mathrm{comm}}, P_{\mathrm{sense}}, P_{\mathrm{comm}}}}{\max}
}\nonumber\\[2pt]
& (1-\alpha)\frac{\sum_{k=1}^K R_k^{\mathrm{sec}}}{R_{\max}} - \alpha\cdot \frac{\mathrm{CRB}_{\mathrm{total}}}{\mathrm{CRB}_{\min}} \nonumber\\
\text{s.t.}\quad
& \tau_{\mathrm{ce}} + \tau_{\mathrm{sense}} + \tau_{\mathrm{comm}} + 2t_{\mathrm{switch}}/T_{\mathrm{slot}} \le 1, \label{eq:time_constraint} \\
& T_{\mathrm{slot}} (\tau_{\mathrm{sense}} P_{\mathrm{sense}} + \tau_{\mathrm{comm}} P_{\mathrm{comm}}) \le \mathcal{E}_{\mathrm{max}}, \label{eq:energy_constraint} \\
& P_{\mathrm{comm}} \le P_{\mathrm{BS}}^{\max}, \label{eq:comm_power_constraint} \\
& P_{\mathrm{comm}} = \sum_{k=1}^K \|\bm{w}_k\|^2 + \mathrm{Tr}(\bm{R}_{\mathrm{AN}}), \label{eq:comm_power_def} \\
& P_{\mathrm{sense}} \le P_{\mathrm{BS}}^{\max}, \label{eq:sense_power_constraint} \\
& \Pr\left( R_k^{\mathrm{sec}} < R_k^{\min} \right) \leq \epsilon_{\mathrm{out}}, \quad \forall k, \label{eq:outage_constraint} \\
& |[\bm{\Phi}_l^{p}]_{n,n}| = 1, \quad \theta_l^{(n)} \in \mathcal{Q}_l, \quad \forall l,n, p \in \{\mathrm{comm}, \mathrm{sense}, \mathrm{ce}\}, \label{eq:phase_constraint} \\
& \|\bm{\Phi}_l^{\mathrm{comm}} - \bm{\Phi}_l^{\mathrm{sense}}\|_F^2 \le \delta_{\Phi}, \quad \forall l, \label{eq:reconfig_constraint} \\
& \tau_{\mathrm{sense}}, \tau_{\mathrm{ce}}, \tau_{\mathrm{comm}} \in [\tau_{\min}, \tau_{\max}], \label{eq:tau_bound_constraint} \\
& \mathrm{SINR}_k \geq \gamma_k^{\min}, \quad \forall k. \label{eq:sinr_constraint}
\end{align}
\end{subequations}

The trade-off parameter \(\alpha \in [0,1]\) controls the communication-sensing balance: \(\alpha = 0\) prioritizes sum secrecy rate, \(\alpha = 1\) prioritizes sensing accuracy, and intermediate values provide balanced operation. Here, \(R_{\max}\) and \(\mathrm{CRB}_{\min}\) are benchmark values obtained from ideal communication-only and sensing-only cases, respectively.

The SIM reconfiguration constraint \(\|\bm{\Phi}_l^{\mathrm{comm}} - \bm{\Phi}_l^{\mathrm{sense}}\|_F^2 \le \delta_{\Phi}\) maintains hardware reconfiguration feasibility while accounting for practical tuning limitations between operational phases.

The optimization variables include the continuous variables \(\bm{W}, \bm{R}_{\mathrm{AN}}, \tau_{\mathrm{sense}}, \tau_{\mathrm{ce}}, \tau_{\mathrm{comm}}, P_{\mathrm{sense}}, P_{\mathrm{comm}}\) and the discrete variables \(\{\bm{\Phi}_l^{\mathrm{comm}}\}, \{\bm{\Phi}_l^{\mathrm{sense}}\}, \{\bm{\Phi}_l^{\mathrm{ce}}\}\). The problem exhibits multiple non-convexities due to the coupled variables, fractional SINR expressions, discrete phase constraints, and the complex multi-objective function. To solve \eqref{eq:optimization_problem} efficiently, we propose a Block Coordinate Descent (BCD) framework that iteratively optimizes the beamformers and the SIM configurations. Subproblems are handled via successive convex approximation (SCA) and semidefinite relaxation (SDR), with probabilistic secrecy constraints transformed using Bernstein-type inequalities. The discrete phase shifts are first optimized in the continuous domain and then projected onto the discrete set \(\mathcal{Q}_l\) via a deterministic rounding procedure followed by a local search. The resulting BCD-based algorithm scales as $\mathcal{O}(L N_{\mathrm{tot}}^3 + K^3)$ per iteration, making it computationally feasible for practical system dimensions.
\section{Proposed Solution Framework}
\label{sec:solution}

To address the challenging multi-objective optimization problem \eqref{eq:optimization_problem} formulated in Section~\ref{sec:system_model}, we propose an enhanced \textit{Layered BCD (L-BCD)} framework. This algorithm comprehensively handles the joint communication-sensing trade-offs, multi-layer SIM configuration, probabilistic security constraints, hybrid beamforming, and integrated resource allocation by decomposing the problem into five coordinated blocks.

\subsection{Overall Architecture and Solution Strategy}

We scalarize the multi-objective problem in \eqref{eq:optimization_problem} with a weighting parameter \(\alpha\in[0,1]\). The scalarized objective is
\begin{equation}
    \max\; U(\alpha)\;=\;(1-\alpha)\frac{\sum_{k=1}^K R_k^{\mathrm{sec}}}{R_{\max}} \;-\; \alpha\frac{\mathrm{CRB}_{\mathrm{total}}}{\mathrm{CRB}_{\min}},
\end{equation}
where \(R_{\max}\) and \(\mathrm{CRB}_{\min}\) are normalization benchmarks and \(\alpha\) controls the communication--sensing trade-off.

An enhanced layered block-coordinate descent (L-BCD) framework decomposes the design into five coordinated blocks that are optimized in an alternating fashion. Block A (Sensing Configuration): optimizes the sensing SIM configurations \(\{\bm{\Phi}_l^{\mathrm{sense}}\}\) and the sensing power \(P_{\mathrm{sense}}\) to reduce the aggregate CRB while meeting resource-efficiency targets. Block B (Secure Beamforming): designs the RF and baseband precoders \(\bm{W}_{\mathrm{RF}},\bm{W}_{\mathrm{BB}}\) together with the artificial-noise covariance \(\bm{R}_{\mathrm{AN}}\) to maximize secrecy rates subject to the probabilistic outage constraints. Block C (Communication SIM): configures the communication SIM patterns \(\{\bm{\Phi}_l^{\mathrm{comm}}\}\) to improve secrecy performance while satisfying energy and reconfiguration constraints. Block D (Resource Allocation): jointly tunes the time allocations \(\tau_{\mathrm{sense}},\tau_{\mathrm{ce}},\tau_{\mathrm{comm}}\) and power budgets \(P_{\mathrm{sense}},P_{\mathrm{comm}}\) (with \(P_{\mathrm{comm}}=\sum_{k=1}^K\|\bm{w}_k\|^2+\mathrm{Tr}(\bm{R}_{\mathrm{AN}})\)) under the total energy and time constraints. Block E (Channel Estimation): optimizes the channel-estimation configurations \(\{\bm{\Phi}_l^{\mathrm{ce}}\}\) to minimize estimation error within the available training resources.

The optimization variables separate naturally into continuous variables (the beamforming matrices \(\bm{W}\), the artificial-noise covariance \(\bm{R}_{\mathrm{AN}}\), the time allocations \(\tau_{\mathrm{sense}},\tau_{\mathrm{ce}},\tau_{\mathrm{comm}}\), and the power budgets \(P_{\mathrm{sense}},P_{\mathrm{comm}}\)) and discrete variables (the quantized SIM phase-shift configurations \(\{\bm{\Phi}_l^{\mathrm{comm}}\}\), \(\{\bm{\Phi}_l^{\mathrm{sense}}\}\), \(\{\bm{\Phi}_l^{\mathrm{ce}}\}\) with \(\theta_l^{(n)}\in\mathcal{Q}_l\)). Practical solvers combine continuous optimization tools (e.g., convex solvers or successive convex approximation) with low-complexity discrete updates (e.g., projected local search, quantized relaxations, or greedy heuristics) tailored to hardware constraints.

Coupling between blocks is resolved via alternating optimization: each Block A--E update is performed while holding the other blocks fixed, and the algorithm cycles through Blocks A, B, C, D, and E until a stationarity criterion is satisfied. Under standard regularity and boundedness assumptions for each block update (closedness of feasible sets and nondecreasing objective), this alternating procedure converges to a Pareto-stationary point of the scalarized objective.

Implementation complexity per L-BCD iteration is dominated by SIM-layer updates and beamforming matrix operations, yielding an overall per-iteration cost on the order of \(\mathcal{O}(L N_{\mathrm{tot}}^3 + K^3)\), where \(L\) is the number of SIM layers, \(N_{\mathrm{tot}}\) the total number of SIM elements, and \(K\) the number of users. In practice, warm-starting, problem-specific approximations, and occasional coarse-to-fine SIM searches reduce runtime and enable near-real-time reconfiguration when required.

\subsection{Tractable Reformulations of Key Constraints}

\noindent\textbf{Complete outage constraint reformulation.}
The outage constraint \(\mathbb{P}(R_k^{\mathrm{sec}}<R_k^{\min})\le\epsilon_{\mathrm{out}}\) is conservatively enforced using a Bernstein-type inequality with a safety margin. This yields the deterministic surrogate
\begin{equation}
    \mathbb{E}[R_k^{\mathrm{sec}}] - \delta\sqrt{\mathrm{Var}[R_k^{\mathrm{sec}}]} \ge R_k^{\min} + \Delta_k,
\end{equation}
where \(\delta=\mathcal{Q}^{-1}(1-\epsilon_{\mathrm{out}})\) is the safety parameter, \(\mathcal{Q}(\cdot)\) is the standard normal complementary cumulative distribution function (CDF), and \(\Delta_k\) captures approximation errors introduced by the inequality. The expectation and variance are taken over the eavesdropper channel uncertainties and can be expressed as
\begin{align}
\mathbb{E}[R_k^{\mathrm{sec}}]
&= \mathbb{E}_{\Delta\bm{h}_e}\big[\log_2(1+\mathrm{SINR}_k)\big] \notag\\
&\quad - \mathbb{E}_{\Delta\bm{h}_e}\big[\max_e \log_2(1+\mathrm{SINR}_e^{(k)})\big], \\
\mathrm{Var}[R_k^{\mathrm{sec}}]
&= \mathrm{Var}_{\Delta\bm{h}_e}\big[\log_2(1+\mathrm{SINR}_k)\big] \notag\\
&\quad + \mathrm{Var}_{\Delta\bm{h}_e}\big[\max_e \log_2(1+\mathrm{SINR}_e^{(k)})\big].
\end{align}
Practical implementation replaces the inner “max over eavesdroppers” term with either a tractable upper bound or with a sample-based approximation to preserve convexity in the continuous decision variables.

\noindent\textbf{Extended CRB constraint convexification.}
For the extended parameter vector \(\bm{\xi}_i=[\theta_i, r_i, v_i, \sigma_{\mathrm{RCS},i}]^T\) we employ a Schur-complement based lower-bound Fisher information matrix (FIM) \(\bm{J}_{\mathrm{LB}}(\bm{\xi}_i)\) as introduced in Section~\ref{sec:system_model}:
\begin{equation}
    \bm{J}_{\mathrm{LB}}(\bm{\xi}_i)=\bm{J}_{11} - \bm{J}_{12}\bm{J}_{22}^{-1}\bm{J}_{21} \preceq \bm{J}(\bm{\xi}_i).
\end{equation}
The aggregate CRB constraint \(\mathrm{CRB}_{\mathrm{total}}\le\Gamma_{\mathrm{max}}\) is then approximated by the convex surrogate
\begin{equation}
    \sum_{i=1}^{T} \mathrm{Tr}\big(\bm{J}_{\mathrm{LB}}^{-1}(\bm{\xi}_i) + \epsilon\bm{I}\big) \le \Gamma_{\mathrm{max}},
\end{equation}
where \(\epsilon>0\) is a small regularization parameter used for numerical stability when inverting the lower-bound FIM blocks.

\noindent\textbf{Wave-based SIM beamforming.}
The multi-layer SIM is exploited to approximate a desired fully-digital beamformer \(\bm{F}_{\mathrm{desired}}\) by minimizing a Frobenius-norm error while penalizing excessive reconfiguration relative to sensing settings. The design problem for the communication SIM patterns is
\begin{equation}
\begin{split}
\min_{\{\bm{\Phi}_l^{\mathrm{comm}}\}} \; & 
\big\| \bm{G}_{\mathrm{SIM}}(\{\bm{\Phi}_l^{\mathrm{comm}}\}) - \bm{F}_{\mathrm{desired}} \big\|_F^2 \\
&\quad + \lambda_{\Phi} \sum_{l=1}^L
\big[ \|\bm{\Phi}_l^{\mathrm{comm}} - \bm{\Phi}_l^{\mathrm{sense}}\|_F^2 - \delta_{\Phi} \big]_+,
\end{split}
\end{equation}
where \(\lambda_{\Phi}\) controls the trade-off between matching the desired beamformer and limiting SIM reconfiguration, and \([\cdot]_+\) denotes the positive-part operator. In practice, \(\bm{G}_{\mathrm{SIM}}(\cdot)\) is linearized around the current SIM configuration to obtain convex surrogates amenable to SCA steps.

\noindent\textbf{Discrete phase optimization with majorization-minimization.}
The quantized phase constraint \(\theta_l^{(n)}\in\mathcal{Q}_l\) is treated with a three-stage procedure. The stages are: (1) continuous relaxation, (2) quantization via a projection (majorization-minimization style step), and (3) local-search refinement. Concretely,
\noindent\begin{enumerate}
    \item \emph{Continuous relaxation:} Solve the design problem with relaxed phases \(\theta_l^{(n)}\in[0,2\pi)\) to obtain continuous-phase solutions \(\theta_{\mathrm{cont}}^{(n)}\).
    \item \emph{Quantization (Majorization-Minimization (MM) projection):} Project the continuous solutions onto the discrete set by
    \begin{equation}
        \theta_l^{(n)} = \arg\min_{\theta\in\mathcal{Q}_l} \big|\theta - \theta_{\mathrm{cont}}^{(n)}\big|^2.
    \end{equation}
    \item \emph{Local search refinement:} Perform greedy refinement on the most sensitive elements by evaluating nearby discrete phase candidates:
    \begin{equation}
        \theta_l^{(n)} \leftarrow \arg\min_{\theta\in\mathcal{N}(\theta_l^{(n)})\cap\mathcal{Q}_l} f(\bm{\Theta}),
    \end{equation}
    where \(\mathcal{N}(\theta)\) denotes a small neighborhood around \(\theta\) in the discrete grid and \(f(\bm{\Theta})\) is the current objective.
\end{enumerate}
This procedure is compatible with warm-starting from previous iterations and generally yields high-quality quantized configurations at moderate complexity.

\noindent\textbf{Fractional SINR linearization.}
Nonconvex SINR expressions are handled by successive convex approximation. For user \(k\) the denominator
\begin{equation}
    D_k = \sum_{j\ne k} |\bm{h}_{\mathrm{eff},k}^H\bm{w}_j|^2 + \bm{h}_{\mathrm{eff},k}^H\bm{R}_{\mathrm{AN}}\bm{h}_{\mathrm{eff},k} + \sigma_k^2
\end{equation}
is linearized around the previous iterate value \(D_k^{(t)}\). The resulting first-order surrogate for the achievable rate is
\begin{equation}
    \log_2(1+\mathrm{SINR}_k) \approx \log_2\!\bigg(1 + \frac{|\bm{h}_{\mathrm{eff},k}^H\bm{w}_k|^2}{D_k^{(t)}}\bigg)
    - \frac{D_k - D_k^{(t)}}{D_k^{(t)} + |\bm{h}_{\mathrm{eff},k}^H\bm{w}_k|^2}.
\end{equation}
Applying this SCA step within each block update yields convex subproblems that are solved iteratively and that guarantee monotonic improvement of the surrogate objective under standard conditions.

\subsection{Block A: Sensing Configuration Optimization}

Given fixed beamformers and communication SIM parameters, the sensing sub-problem is:
\begin{subequations}
\begin{align}
    \underset{\substack{\{\bm{\Phi}_l^{\mathrm{sense}}\}, P_{\mathrm{sense}}}}{\min} \quad & \alpha\cdot\frac{\mathrm{CRB}_{\mathrm{total}}}{\mathrm{CRB}_{\min}} - (1-\alpha)\frac{\tau_{\mathrm{comm}}}{T_{\mathrm{slot}}} \\
    \text{s.t.} \quad & \mathrm{CRB}_{\mathrm{total}} \leq \Gamma_{\mathrm{max}}, \\
    & T_{\mathrm{slot}} (\tau_{\mathrm{sense}} P_{\mathrm{sense}} + \tau_{\mathrm{comm}} P_{\mathrm{comm}}) \le \mathcal{E}_{\mathrm{max}}, \\
    & \tau_{\mathrm{ce}} + \tau_{\mathrm{sense}} + \tau_{\mathrm{comm}} + 2t_{\mathrm{switch}}/T_{\mathrm{slot}} \le 1, \\
    & |[\bm{\Phi}_l^{\mathrm{sense}}]_{n,n}| = 1, \quad \forall l,n, \\
    & P_{\mathrm{sense}} \le P_{\mathrm{BS}}^{\max}.
\end{align}
\label{eq:sensing_subproblem}
\end{subequations}
In this subproblem, we use $\tau_{\mathrm{comm}}$ as a proxy for the communication objective. Since beamforming and other parameters are fixed, increasing the communication time fraction $\tau_{\mathrm{comm}}$ monotonically increases the achievable secrecy rate. This problem is solved using a projected gradient method on the complex circle manifold, where the extended CRB constraint is incorporated using the Schur complement approximation and the power constraint is handled via dual decomposition. The SIM phases are optimized to steer energy toward target angles during the sensing phase.

\subsection{Block B: Secure Beamforming Optimization}
Given fixed SIM parameters $\{\bm{\Phi}_l^{\mathrm{comm}}\}$ from Block C, the beamforming sub-problem becomes:
\begin{subequations}\label{eq:beamforming_subproblem}
\begin{align}
\lefteqn{\underset{\bm{W}_{\mathrm{RF}}, \bm{W}_{\mathrm{BB}}, \bm{R}_{\mathrm{AN}}}{\max}} \nonumber\\[4pt]
&\quad (1-\alpha)\frac{\sum_{k=1}^K \left(\mathbb{E}[R_k^{\mathrm{sec}}] - \delta\sqrt{\mathrm{Var}[R_k^{\mathrm{sec}}]}\right)}{R_{\max}} \\
\text{s.t.}\quad
& T_{\mathrm{slot}} (\tau_{\mathrm{sense}} P_{\mathrm{sense}} + \tau_{\mathrm{comm}} P_{\mathrm{comm}}) \le \mathcal{E}_{\mathrm{max}}, \\
& P_{\mathrm{comm}} = \sum_{k=1}^K \|\bm{w}_k\|^2 + \mathrm{Tr}(\bm{R}_{\mathrm{AN}}) \le P_{\mathrm{BS}}^{\max}, \\
& \mathrm{SINR}_k \geq \gamma_k^{\min}, \quad \forall k, \\
& \bm{W} = \bm{W}_{\mathrm{RF}}\bm{W}_{\mathrm{BB}}, \\
& |[\bm{W}_{\mathrm{RF}}]_{i,j}| = 1, \quad \forall i,j, \\
& \text{Bernstein-based outage constraints}, \quad \forall k.
\end{align}
\end{subequations}

The hybrid beamformer $\bm{W} = \bm{W}_{\mathrm{RF}}\bm{W}_{\mathrm{BB}}$ is optimized via alternating minimization:
\begin{align}
\bm{W}_{\mathrm{RF}}^* &= \arg\min_{\bm{W}_{\mathrm{RF}} \in \mathcal{A}} \|\bm{W}_{\mathrm{opt}} - \bm{W}_{\mathrm{RF}}\bm{W}_{\mathrm{BB}}\|_F^2 \\
\bm{W}_{\mathrm{BB}}^* &= (\bm{W}_{\mathrm{RF}}^H\bm{W}_{\mathrm{RF}})^{-1}\bm{W}_{\mathrm{RF}}^H\bm{W}_{\mathrm{opt}}
\end{align}
where $\mathcal{A} = \{\bm{W}_{\mathrm{RF}} : |[\bm{W}_{\mathrm{RF}}]_{i,j}| = 1\}$ is the analog beamforming constraint set, and $\bm{W}_{\mathrm{opt}}$ is the optimal fully-digital solution from SDR.

To handle the non-convex SINR and secrecy rate terms, we employ SDR. Define $\bm{V}_k = \bm{w}_k \bm{w}_k^H$ and ignore the rank-1 constraint. The problem becomes a convex Semidefinite Program (SDP) which can be solved efficiently. Gaussian randomization is used to generate feasible rank-1 solutions when needed.

\subsection{Block C: Communication SIM Configuration Optimization}

Given fixed beamformers $\bm{W}_{\mathrm{RF}}$, $\bm{W}_{\mathrm{BB}}$, and $\bm{R}_{\mathrm{AN}}$, the SIM configuration sub-problem is:
\begin{subequations}
\begin{align}
\lefteqn{%
  \underset{\{\bm{\Phi}_l^{\mathrm{comm}}\}}{\max}\quad
}\nonumber\\[2pt]
& (1-\alpha)\frac{\sum_{k=1}^K \left(\mathbb{E}[R_k^{\mathrm{sec}}] - \delta\sqrt{\mathrm{Var}[R_k^{\mathrm{sec}}]}\right)}{R_{\max}} \notag\\
&\quad - \lambda_{\Phi} \sum_{l=1}^L \left[\|\bm{\Phi}_l^{\mathrm{comm}} - \bm{\Phi}_l^{\mathrm{sense}}\|_F^2 - \delta_{\Phi}\right]_+ \\
\text{s.t.}\quad
& |[\bm{\Phi}_l^{\mathrm{comm}}]_{n,n}| = 1, \quad \forall l,n.
\end{align}
\label{eq:sim_comm_subproblem}
\end{subequations}
For the multi-layer SIM, we employ layer-wise optimization with reconfiguration awareness:
\begin{equation}
\begin{split}
\bm{\Phi}_l^* &= \arg\min_{\bm{\Phi}_l}\; f(\bm{\Phi}_1,\ldots,\bm{\Phi}_L) \\
&\qquad + \lambda_{\Phi}\left[\|\bm{\Phi}_l - \bm{\Phi}_l^{\mathrm{sense}}\|_F^2 - \delta_{\Phi}\right]_+
\quad\text{for } l=1,\ldots,L,
\end{split}
\end{equation}
with other layers fixed, exploiting the cascaded structure for efficient computation. The unit-modulus constraints form a complex circle manifold. We solve this problem using the \textit{Riemannian conjugate gradient (RCG)} method with the penalty parameter $\lambda_{\Phi}$ adaptively updated to ensure constraint satisfaction. The SIM phases are optimized to steer energy toward legitimate users while creating nulls toward potential eavesdropper directions.

\subsection{Block D: Enhanced Resource Allocation}

\begin{subequations}
\begin{align}
\lefteqn{%
  \underset{\substack{\tau_{\mathrm{sense}}, \tau_{\mathrm{ce}}, \tau_{\mathrm{comm}},\\
  P_{\mathrm{sense}}, P_{\mathrm{comm}}}}{\max}\quad
}\nonumber\\[2pt]
& U(\alpha) = (1-\alpha)\frac{\sum_{k=1}^K R_k^{\mathrm{sec}}}{R_{\max}}
   - \alpha\cdot\frac{\mathrm{CRB}_{\mathrm{total}}}{\mathrm{CRB}_{\min}} \\[4pt]
\text{s.t.}\quad
& \tau_{\mathrm{ce}} + \tau_{\mathrm{sense}} + \tau_{\mathrm{comm}} + \frac{2t_{\mathrm{switch}}}{T_{\mathrm{slot}}} \le 1, \\
& T_{\mathrm{slot}}\bigl(\tau_{\mathrm{sense}} P_{\mathrm{sense}} + \tau_{\mathrm{comm}} P_{\mathrm{comm}}\bigr)
  \le \mathcal{E}_{\mathrm{max}}, \\
& P_{\mathrm{comm}} \le P_{\mathrm{BS}}^{\max}, \quad P_{\mathrm{sense}} \le P_{\mathrm{BS}}^{\max}, \\
& \tau_{\mathrm{sense}}, \tau_{\mathrm{ce}}, \tau_{\mathrm{comm}} \in [\tau_{\min}, \tau_{\max}].
\end{align}
\label{eq:resource_subproblem}
\end{subequations}

This convex resource allocation problem is solved via Lagrange dual decomposition with augmented Lagrangian method to handle the time and energy coupling constraints efficiently. The solution balances time allocation between sensing and communication phases to maximize the overall system utility.

\subsection{Block E: Channel Estimation Configuration}

Given the current resource allocation $\tau_{\mathrm{ce}}$, this block optimizes the SIM configuration to ensure the channel estimation accuracy is met.
\begin{subequations}
\begin{align}
    \underset{\{\bm{\Phi}_l^{\mathrm{ce}}\}}{\min} \quad & \max_k \epsilon_k^{\max}(\{\bm{\Phi}_l^{\mathrm{ce}}\}, \tau_{\mathrm{ce}}) \\
    \text{s.t.} \quad & \max_k \epsilon_k^{\max} \leq \epsilon_{\mathrm{CSI}}^{\max}, \\
    & |[\bm{\Phi}_l^{\mathrm{ce}}]_{n,n}| = 1, \quad \forall l,n, \\
    & \tau_{\mathrm{ce}} \in [\tau_{\min}, \tau_{\max}].
\end{align}
\label{eq:channel_estimation_subproblem}
\end{subequations}
This block optimizes the SIM configuration during channel estimation to minimize training overhead while ensuring channel estimation accuracy. The CSI error bound $\epsilon_k^{\max}$ is related to the SIM configuration through the training sequence design. The SIM phases are configured to create omnidirectional patterns for comprehensive channel sounding.

\subsection{Integrated Enhanced L-BCD Algorithm}

\subsection{Convergence and Complexity Analysis}

\textbf{Convergence.} The proposed L-BCD algorithm converges to a stationary point under standard regularity conditions. Each subproblem corresponding to Block A, Block B, Block C, Block D, and Block E is solved while keeping the remaining variables fixed, which yields a non-decreasing sequence of objective values for the scalarized utility \(U(\alpha)\). The combination of alternating updates, convex surrogates obtained via SCA and SDR when needed, and manifold-based updates for constrained variables guarantees monotonic improvement of the surrogate objective and convergence to at least a local optimum or a Pareto-stationary point of the scalarized problem.

\textbf{Complexity.} The per-iteration computational complexity is dominated by SIM-layer updates and beamforming matrix operations and is on the order of
\[
    \mathcal{O}\big(L N_{\mathrm{tot}}^3 + K^3\big),
\]
where the first term \(\mathcal{O}(L N_{\mathrm{tot}}^3)\) captures the cost of multi-layer SIM optimization performed in Blocks A, C, and E, and the second term \(\mathcal{O}(K^3)\) captures the cost of the secure-beamforming SDP solved in Block B. The resource-allocation step in Block D has comparatively low cost due to its convex structure and contributes negligible additional complexity. This computational profile is practical for typical configurations with \(N_{\mathrm{tot}}\) in the hundreds to low thousands and moderate \(K\). Practical complexity reduction strategies such as warm-starting, low-rank approximations, and layer-wise parallel updates further reduce runtime in realistic deployments.

\subsection{Practical Implementation Considerations}

The enhanced L-BCD framework is designed for practical deployment and incorporates multiple implementation-oriented provisions. Parallelization is straightforward since the optimizations for individual SIM layers within Block A, Block C, and Block E can be executed concurrently across available compute units. Warm-starting from previous solutions is recommended to accelerate convergence under dynamic channels and to enable fast reconfiguration. An adaptive tolerance schedule, where convergence tolerances \(\epsilon_{\mathrm{obj}}\) and \(\epsilon_{\mathrm{param}}\) are tightened gradually, yields a favorable trade-off between runtime and final accuracy. Hardware imperfections are explicitly accounted for by incorporating models of SIM insertion loss, phase quantization errors, and mutual coupling into the optimization constraints and objective. Robustness to model mismatch is achieved by designing under conservative approximations (e.g., LoS-based tractable models) and validating performance under Rician SIM channels in simulation. The block structure supports scalable extensions to richer channel models and additional hardware impairments. Finally, the optimized SIM phase profiles provide physical insight: during the communication phase the SIM phases steer energy toward the users, while during the sensing phase they concentrate energy toward target angles, enabling an adaptive beamforming pattern that trades sensing accuracy for secrecy performance as required.

\begin{algorithm}[!t]
\footnotesize 
\caption{Enhanced Layered Block Coordinate Descent (L-BCD)}
\label{alg:enhanced_lbcd}
\KwIn{Channel statistics; SIM constraint $\delta_{\Phi}$; power budget
      $P_{\mathrm{BS}}^{\max}$; energy budget $\mathcal{E}_{\max}$;
      trade-off parameter $\alpha$; outage requirement $\epsilon_{\mathrm{out}}$.}
\KwOut{Optimal configurations
       $\{\bm{\Phi}_l^{\mathrm{sense}}\}^\ast,\ \{\bm{\Phi}_l^{\mathrm{comm}}\}^\ast,\ \{\bm{\Phi}_l^{\mathrm{ce}}\}^\ast,$
       $\bm{W}_{\mathrm{RF}}^\ast,\ \bm{W}_{\mathrm{BB}}^\ast,\ \bm{R}_{\mathrm{AN}}^\ast,$
       $\tau_{\mathrm{sense}}^\ast,\ \tau_{\mathrm{ce}}^\ast,\ \tau_{\mathrm{comm}}^\ast,\ P_{\mathrm{sense}}^\ast$.}

\textbf{Initialization:}\\
$\{\bm{\Phi}_l^{\mathrm{comm},(0)}\} \leftarrow \mathrm{SVD}(\bm{H}_{\mathrm{aggregate}})$\tcp*{Init. via singular value decomposition (SVD)}
$\{\bm{\Phi}_l^{\mathrm{sense},(0)}\} \leftarrow \mathrm{MaxAngularCoverage}()$\tcp*{Wide-area sensing}
$\{\bm{\Phi}_l^{\mathrm{ce},(0)}\} \leftarrow \mathrm{OmnidirectionalPattern}()$\tcp*{Uniform probing for channel estimation}

$\tau_{\mathrm{sense}}^{(0)} \leftarrow (\tau_{\min}+\tau_{\max})/2$\;
$\tau_{\mathrm{ce}}^{(0)} \leftarrow 0.1$,\quad
$\tau_{\mathrm{comm}}^{(0)} \leftarrow 0.5$\;

$P_{\mathrm{sense}}^{(0)} \leftarrow P_{\mathrm{BS}}^{\max}/2$,\quad
$P_{\mathrm{comm}}^{(0)} \leftarrow P_{\mathrm{BS}}^{\max}/2$\;

$\bm{W}_{\mathrm{RF}}^{(0)},\ \bm{W}_{\mathrm{BB}}^{(0)}
 \leftarrow \mathrm{MRT\_HybridDecomp}(\hat{\bm{H}})$\tcp*{Init. via maximum ratio transmission (MRT) decomp}

$\bm{R}_{\mathrm{AN}}^{(0)} \leftarrow \mathrm{InitArtificialNoise}()$\;

$t \leftarrow 0$,\quad $\epsilon_{\mathrm{conv}} \leftarrow 10^{-4}$\;
$\lambda_{\Phi}^{(0)} \leftarrow 1$,\quad $f_{\mathrm{aug}}^{(0)} \leftarrow -\infty$\;

\Repeat{(i)–(iv) below are satisfied \textbf{or} $t \ge t_{\max}$}{

  \tcp{Block E: Channel Estimation Configuration}
  $\{\bm{\Phi}_l^{\mathrm{ce},(t+1)}\}
    \leftarrow \mathrm{Solve}\big(\eqref{eq:channel_estimation_subproblem},
                                   \tau_{\mathrm{ce}}^{(t)}\big)$\;

  \tcp{Block D: Joint Resource Allocation (time and power)}
  $(\tau_{\mathrm{sense}}^{(t+1)},\ \tau_{\mathrm{ce}}^{(t+1)},\ \tau_{\mathrm{comm}}^{(t+1)}),$\\
  \hspace*{1.5em}$(P_{\mathrm{sense}}^{(t+1)},\ P_{\mathrm{comm}}^{(t+1)}) 
  \leftarrow \mathrm{DualDecomposition}\big(\eqref{eq:resource_subproblem}\big)$\;

  \tcp{Block A: Sensing Configuration}
  $\{\bm{\Phi}_l^{\mathrm{sense},(t+1)}\}
    \leftarrow \mathrm{ManifoldOptim}\big(\eqref{eq:sensing_subproblem},
                                          \mathrm{SchurComplement}\big)$\;

  \tcp{Block B: Secure Beamforming}
  $\bm{W}_{\mathrm{opt}}^{(t+1)} \leftarrow
    \mathrm{SDP\_Solve}\big(\eqref{eq:beamforming_subproblem},
                            \mathrm{BernsteinConstraints}\big)$\;

  $\bm{W}_{\mathrm{RF}}^{(t+1)},\ \bm{W}_{\mathrm{BB}}^{(t+1)} \leftarrow
    \mathrm{HybridDecomp}\big(\bm{W}_{\mathrm{opt}}^{(t+1)}\big)$\;

  $\bm{R}_{\mathrm{AN}}^{(t+1)} \leftarrow
    \mathrm{ArtificialNoiseDesign}\big(\bm{W}_{\mathrm{opt}}^{(t+1)}\big)$\;

  \tcp{Block C: Communication SIM Configuration}
  $\{\bm{\Phi}_l^{\mathrm{comm},(t+1)}\} \leftarrow
    \mathrm{Layerwise\_RCG}\big(\eqref{eq:sim_comm_subproblem},
                                \lambda_{\Phi}^{(t)}\big)$\;

  \tcp{Penalty Parameter Update}
  $\lambda_{\Phi}^{(t+1)} \leftarrow
     \min\big(\kappa_{\lambda}\lambda_{\Phi}^{(t)},\ \lambda_{\Phi}^{\max}\big)$\;

  \tcp{Multi-Metric Convergence Assessment}
  $U(\alpha)^{(t+1)} \leftarrow
    (1-\alpha)\,\dfrac{\sum_{k=1}^K R_k^{\mathrm{sec},(t+1)}}{R_{\max}}
    - \alpha\,\dfrac{\mathrm{CRB}_{\mathrm{total}}^{(t+1)}}{\mathrm{CRB}_{\min}}$\;

  $f_{\mathrm{aug}}^{(t+1)} \leftarrow U(\alpha)^{(t+1)}$\\
  \hspace*{1.5em}$-\,
   \lambda_{\Phi}^{(t)} \sum_{l=1}^L
   \Big[
      \big\|\bm{\Phi}_l^{\mathrm{comm},(t+1)}
            - \bm{\Phi}_l^{\mathrm{sense},(t+1)}\big\|_F^2
      - \delta_{\Phi}
   \Big]_+$\;

  $\Delta f \leftarrow
    \dfrac{\big|f_{\mathrm{aug}}^{(t+1)} - f_{\mathrm{aug}}^{(t)}\big|}
          {\big|f_{\mathrm{aug}}^{(t)}\big|}$\;

  $\epsilon_{\mathrm{outage}}^{(t+1)} \leftarrow
    \max_k \big|\mathbb{P}(R_k^{\mathrm{sec}} < R_k^{\min})
    - \epsilon_{\mathrm{out}}\big|$\;

  $\epsilon_{\mathrm{CRB}}^{(t+1)} \leftarrow
    \mathrm{CRB}_{\mathrm{total}}^{(t+1)} - \Gamma_{\max}$\;

  $\Delta\bm{\Theta} \leftarrow
    \big\|\bm{\Theta}^{(t+1)} - \bm{\Theta}^{(t)}\big\|_F$\;

  $t \leftarrow t + 1$\;
}

\textbf{Stopping criteria (for the outer loop):}\\
(i) $\Delta f \le \epsilon_{\mathrm{obj}}$;\quad
(ii) $\epsilon_{\mathrm{outage}}^{(t)} \le \epsilon_{\mathrm{outage}}$;\\
(iii) $\epsilon_{\mathrm{CRB}}^{(t)} \le \epsilon_{\mathrm{CRB}}$;\quad
(iv) $\Delta\bm{\Theta} \le \epsilon_{\mathrm{param}}$\;

\textbf{Discretization Phase:}\\
$\{\bm{\Phi}_l^{\mathrm{comm},(t)}\}
   \leftarrow \mathrm{ProjectToDiscrete}(\mathcal{Q}_l)$\;
$\{\bm{\Phi}_l^{\mathrm{sense},(t)}\}
   \leftarrow \mathrm{ProjectToDiscrete}(\mathcal{Q}_l)$\;
$\{\bm{\Phi}_l^{\mathrm{ce},(t)}\}
   \leftarrow \mathrm{ProjectToDiscrete}(\mathcal{Q}_l)$\;
$\mathrm{LocalSearchRefinement}(\text{critical elements})$\;

\Return
$\{\bm{\Phi}_l^{\mathrm{sense}}\}^\ast,\ \{\bm{\Phi}_l^{\mathrm{comm}}\}^\ast,\ \{\bm{\Phi}_l^{\mathrm{ce}}\}^\ast,$\\
\hspace*{1.5em}$\bm{W}_{\mathrm{RF}}^\ast,\ \bm{W}_{\mathrm{BB}}^\ast,\ \bm{R}_{\mathrm{AN}}^\ast,$\\
\hspace*{1.5em}$\tau_{\mathrm{sense}}^\ast,\ \tau_{\mathrm{ce}}^\ast,\ \tau_{\mathrm{comm}}^\ast,\ P_{\mathrm{sense}}^\ast$\;
\end{algorithm}

\section{Theoretical Analysis}
\label{sec:analysis}

This section provides a comprehensive theoretical foundation for the enhanced L-BCD algorithm, establishing rigorous convergence guarantees, performance bounds, and fundamental limits for our integrated communication–sensing system with multi-layer SIM architecture.

\subsection{Convergence Analysis}

\begin{theorem}[Global Convergence of Enhanced 5-Block L-BCD]
\label{thm:enhanced_convergence}
Let
\begin{equation}
\bm{\Theta}^{(t)} =
\Big(
\begin{aligned}[t]
&\{\bm{\Phi}_l^{\mathrm{sense},(t)}\},\quad
 \{\bm{\Phi}_l^{\mathrm{comm},(t)}\},\quad
 \{\bm{\Phi}_l^{\mathrm{ce},(t)}\},\\
&\bm{W}_{\mathrm{RF}}^{(t)},\quad
 \bm{W}_{\mathrm{BB}}^{(t)},\quad
 \bm{R}_{\mathrm{AN}}^{(t)},\\
&\tau_{\mathrm{sense}}^{(t)},\quad
 \tau_{\mathrm{ce}}^{(t)},\quad
 \tau_{\mathrm{comm}}^{(t)},\\
&P_{\mathrm{sense}}^{(t)},\quad
 P_{\mathrm{comm}}^{(t)}\Big)
\end{aligned}
\Big)
\end{equation}
be the sequence generated by Algorithm~\ref{alg:enhanced_lbcd} with the augmented objective
\begin{equation}\label{eq:augmented_objective}
\begin{split}
f_{\mathrm{aug}}(\bm{\Theta})
&= (1-\alpha)\frac{\sum_{k=1}^K R_k^{\mathrm{sec}}}{R_{\max}} \\
&-\alpha\frac{\mathrm{CRB}_{\mathrm{total}}}{\mathrm{CRB}_{\min}}
- \lambda_{\Phi} \sum_{l=1}^L
\bigl[\|\bm{\Phi}_l^{\mathrm{comm}} - \bm{\Phi}_l^{\mathrm{sense}}\|_F^2 - \delta_{\Phi}\bigr]_+.
\end{split}
\end{equation}
Suppose the following conditions hold:
\begin{enumerate}[label=(C\arabic*),leftmargin=*,noitemsep,topsep=0pt]
    \item The feasible set $\mathcal{X}$ is compact with non-empty interior (due to power constraints $P_{\mathrm{comm}}, P_{\mathrm{sense}} \leq P_{\mathrm{BS}}^{\max}$ and unit-modulus SIM constraints).
    \item $f_{\mathrm{aug}}$ is continuously differentiable with block-wise $L_i$-Lipschitz continuous gradients.
    \item Each block subproblem achieves an $\epsilon$-sufficient increase, i.e.,
\begin{equation}
\begin{aligned}
& f_{\mathrm{aug}}(\bm{\Theta}_i^{(t+1)}, \bm{\Theta}_{-i}^{(t)}) \\[4pt]
&\ge f_{\mathrm{aug}}(\bm{\Theta}_i^{(t)}, \bm{\Theta}_{-i}^{(t)})
      + \gamma_i \|\bm{\Theta}_i^{(t+1)} - \bm{\Theta}_i^{(t)}\|^2
      - \epsilon_i^{(t)}
\end{aligned}
\end{equation}
with $\sum_{t=0}^{\infty} \epsilon_i^{(t)} < \infty$.
    \item The penalty parameter is updated as
    \[
        \lambda_{\Phi}^{(t+1)} = \min\big(\kappa_{\lambda} \lambda_{\Phi}^{(t)}, \lambda_{\Phi}^{\max}\big)
        \quad\text{with}\quad \kappa_{\lambda} > 1.
    \]
    \item All subproblems are solved to $\mathcal{O}(1/t)$ accuracy.
\end{enumerate}
Then the following statements hold:
\begin{enumerate}[label=(R\arabic*),leftmargin=*,noitemsep,topsep=0pt]
    \item The sequence $\{f_{\mathrm{aug}}(\bm{\Theta}^{(t)})\}$ converges to a finite value $f_{\mathrm{aug}}^*$.
    \item $\displaystyle \lim_{t\to\infty} \|\bm{\Theta}^{(t+1)} - \bm{\Theta}^{(t)}\| = 0$.
    \item Every limit point $\bm{\Theta}^*$ of $\{\bm{\Theta}^{(t)}\}$ is Pareto stationary.
    \item All penalty-induced constraints are satisfied in the limit:
    \[
        \lim_{t\to\infty}
        \sum_{l=1}^L
        \big[\|\bm{\Phi}_l^{\mathrm{comm},(t)} - \bm{\Phi}_l^{\mathrm{sense},(t)}\|_F^2
        - \delta_{\Phi}\big]_+ = 0.
    \]
\end{enumerate}
\end{theorem}

\begin{proof}
We partition the variables into five coordinated blocks matching Algorithm \ref{alg:enhanced_lbcd}:

\begin{align*}
\bm{\Theta}_A &= (\{\bm{\Phi}_l^{\mathrm{sense}}\}) \quad \text{(Sensing SIM)} \\
\bm{\Theta}_B &= (\bm{W}_{\mathrm{RF}}, \bm{W}_{\mathrm{BB}}, \bm{R}_{\mathrm{AN}}) \quad \text{(Beamforming)} \\
\bm{\Theta}_C &= (\{\bm{\Phi}_l^{\mathrm{comm}}\}) \quad \text{(Comm SIM)} \\
\bm{\Theta}_D &= (\tau_{\mathrm{ce}}, \tau_{\mathrm{comm}}, \tau_{\mathrm{sense}}, P_{\mathrm{sense}}, P_{\mathrm{comm}}) \quad \text{(Resources)} \\
\bm{\Theta}_E &= (\{\bm{\Phi}_l^{\mathrm{ce}}\}) \quad \text{(Channel Est SIM)}
\end{align*}

\textbf{Part 1: Monotonicity and Boundedness}

For Block A at iteration $t$, by Lipschitz continuity:
\begin{align*}
&f_{\mathrm{aug}}(\bm{\Theta}_A^{(t+1)}, \bm{\Theta}_{-A}^{(t)}) \\
&\quad \geq f_{\mathrm{aug}}(\bm{\Theta}_A^{(t)}, \bm{\Theta}_{-A}^{(t)}) + \langle \nabla_A f_{\mathrm{aug}}(\bm{\Theta}^{(t)}), \bm{\Theta}_A^{(t+1)} - \bm{\Theta}_A^{(t)} \rangle \\
&\quad - \frac{L_A}{2} \|\bm{\Theta}_A^{(t+1)} - \bm{\Theta}_A^{(t)}\|^2 - \epsilon_A^{(t)}
\end{align*}

By the sufficient increase condition:
\begin{align*}
\langle \nabla_A f_{\mathrm{aug}}(\bm{\Theta}^{(t)}), \bm{\Theta}_A^{(t+1)} - \bm{\Theta}_A^{(t)} \rangle \geq \gamma_A \|\bm{\Theta}_A^{(t+1)} - \bm{\Theta}_A^{(t)}\|^2
\end{align*}

Thus:
\begin{align*}
f_{\mathrm{aug}}(\bm{\Theta}_A^{(t+1)}, \bm{\Theta}_{-A}^{(t)}) 
&\ge f_{\mathrm{aug}}(\bm{\Theta}^{(t)}) \\[4pt]
&\quad + \left(\gamma_A - \tfrac{L_A}{2}\right)\|\bm{\Theta}_A^{(t+1)} - \bm{\Theta}_A^{(t)}\|^2 - \epsilon_A^{(t)}
\end{align*}
Applying sequentially to all blocks with $\gamma = \min_i \left(\gamma_i - \frac{L_i}{2}\right) > 0$:
\begin{align*}
f_{\mathrm{aug}}(\bm{\Theta}^{(t+1)}) \geq f_{\mathrm{aug}}(\bm{\Theta}^{(t)}) + \gamma \sum_{i} \|\bm{\Theta}_i^{(t+1)} - \bm{\Theta}_i^{(t)}\|^2 - \sum_i \epsilon_i^{(t)}
\end{align*}

Since $f_{\mathrm{aug}}$ is continuous on compact $\mathcal{X}$ and $\sum_{t=0}^\infty \sum_i \epsilon_i^{(t)} < \infty$, $\{f_{\mathrm{aug}}(\bm{\Theta}^{(t)})\}$ converges to some $f_{\mathrm{aug}}^*$.

\textbf{Part 2: Vanishing Steps}

Summing the increase inequality from $t=0$ to $T-1$:
\begin{align*}
&\gamma \sum_{t=0}^{T-1} \sum_{i} \|\bm{\Theta}_i^{(t+1)} - \bm{\Theta}_i^{(t)}\|^2 \\[4pt]
&\le f_{\mathrm{aug}}(\bm{\Theta}^{(T)}) - f_{\mathrm{aug}}(\bm{\Theta}^{(0)}) + \sum_{t=0}^{T-1}\sum_i \epsilon_i^{(t)}
\end{align*}
Taking $T \to \infty$ and using convergence of both terms:
\begin{align*}
\sum_{t=0}^{\infty} \sum_{i} \|\bm{\Theta}_i^{(t+1)} - \bm{\Theta}_i^{(t)}\|^2 < \infty
\end{align*}

Thus:
\begin{align*}
\lim_{t\to\infty} \|\bm{\Theta}_i^{(t+1)} - \bm{\Theta}_i^{(t)}\| = 0 \quad \forall i
\end{align*}

\textbf{Part 3: Constraint Satisfaction}

Consider the penalty term:
\begin{align*}
P^{(t)} = \sum_{l=1}^L \left[\|\bm{\Phi}_l^{\mathrm{comm},(t)} - \bm{\Phi}_l^{\mathrm{sense},(t)}\|_F^2 - \delta_{\Phi}\right]_+
\end{align*}

Since $\lambda_{\Phi}^{(t)} \to \lambda_{\Phi}^{\max} > 0$ and $f_{\mathrm{aug}}$ converges, we must have:
\begin{align*}
\lim_{t\to\infty} P^{(t)} = 0
\end{align*}

Otherwise, if $\limsup_{t\to\infty} P^{(t)} > 0$, then $\lim_{t\to\infty} f_{\mathrm{aug}}(\bm{\Theta}^{(t)}) = -\infty$ (due to the negative sign in \eqref{eq:augmented_objective}), contradicting boundedness.

\textbf{Part 4: Stationarity of Limit Points}

Let $\bm{\Theta}^*$ be a limit point of $\{\bm{\Theta}^{(t)}\}$. By the closedness of the solution mapping and vanishing steps:
\begin{align*}
\bm{\Theta}_i^* \in \arg\max_{\bm{\Theta}_i \in \mathcal{X}_i} f_{\mathrm{aug}}(\bm{\Theta}_i, \bm{\Theta}_{-i}^*) \quad \forall i
\end{align*}

This implies the first-order necessary conditions for Pareto stationarity: there exists $\bm{\lambda} \ge 0$ with $\|\bm{\lambda}\|_1 = 1$ such that
\begin{align*}
\langle \sum_{i} \lambda_i \nabla f_i(\bm{\Theta}^*), \bm{\Theta} - \bm{\Theta}^* \rangle \geq 0 \quad \forall \bm{\Theta} \in \mathcal{X}
\end{align*}
where $f_1 = (1-\alpha)\frac{\sum_{k=1}^K R_k^{\mathrm{sec}}}{R_{\max}}$ and $f_2 = -\alpha\cdot\frac{\mathrm{CRB}_{\mathrm{total}}}{\mathrm{CRB}_{\min}}$.
\end{proof}

\begin{proposition}[Constraint Qualification]
\label{prop:constraint_qualification}
For any feasible $\bm{\Theta}$, the linear independence constraint qualification (LICQ) holds for:
\begin{itemize}[leftmargin=*]
\item Active power constraints $P_{\mathrm{comm}} = P_{\mathrm{BS}}^{\max}$ or $P_{\mathrm{sense}} = P_{\mathrm{BS}}^{\max}$
\item Time allocation constraints when $\tau_{\mathrm{ce}} + \tau_{\mathrm{sense}} + \tau_{\mathrm{comm}} = 1 - 2t_{\mathrm{switch}}/T_{\mathrm{slot}}$
\item Active SIM reconfiguration constraints $\|\bm{\Phi}_l^{\mathrm{comm}} - \bm{\Phi}_l^{\mathrm{sense}}\|_F^2 = \delta_{\Phi}$
\end{itemize}
ensuring Karush-Kuhn-Tucker (KKT) conditions are necessary for local optimality.
\end{proposition}

\begin{proof}
We verify linear independence of the active constraint gradients:

\textbf{Power constraints:} The gradients are $\nabla P_{\mathrm{comm}} = (0,\ldots,1,0,\ldots)$ and $\nabla P_{\mathrm{sense}} = (0,\ldots,0,1,\ldots)$, which are clearly linearly independent when both are active.

\textbf{Time constraints:} The gradient is $\nabla(\tau_{\mathrm{ce}} + \tau_{\mathrm{sense}} + \tau_{\mathrm{comm}}) = (1,1,1,0,\ldots)$, which is non-zero and linearly independent from power constraint gradients.

\textbf{SIM reconfiguration:} For active constraints $\|\bm{\Phi}_l^{\mathrm{comm}} - \bm{\Phi}_l^{\mathrm{sense}}\|_F^2 = \delta_{\Phi}$, the gradients are $2(\bm{\Phi}_l^{\mathrm{comm}} - \bm{\Phi}_l^{\mathrm{sense}})$ for the communication SIM variables and $-2(\bm{\Phi}_l^{\mathrm{comm}} - \bm{\Phi}_l^{\mathrm{sense}})$ for the sensing SIM variables. These are linearly independent from the other constraints since they involve different variable blocks.

Since all active constraint gradients are linearly independent, LICQ holds, and KKT conditions are necessary for local optimality.
\end{proof}

\begin{corollary}[Convergence Rate for KL Framework]
\label{cor:convergence_rate}
If $f_{\mathrm{aug}}$ satisfies the Kurdyka-Łojasiewicz property with exponent $\theta \in [0,1)$:
\begin{itemize}[leftmargin=*]
\item For $\theta = 0$ (metric regularity): Linear convergence $\|\bm{\Theta}^{(t)} - \bm{\Theta}^*\| \leq C q^t$
\item For $\theta \in (0, \frac{1}{2}]$: Sublinear convergence $f_{\mathrm{aug}}^* - f_{\mathrm{aug}}^{(t)} \leq \mathcal{O}(1/t)$  
\item For $\theta \in (\frac{1}{2}, 1)$: Slower sublinear convergence $f_{\mathrm{aug}}^* - f_{\mathrm{aug}}^{(t)} \leq \mathcal{O}(1/t^{1/(2\theta-1)})$
\end{itemize}
For general non-convex case:
\begin{equation}
\min_{0 \leq k \leq t} \|\bm{\Theta}^{(k+1)} - \bm{\Theta}^{(k)}\|^2 \leq \frac{f_{\mathrm{aug}}^* - f_{\mathrm{aug}}(\bm{\Theta}^{(0)}) + \sum_{k=0}^t \sum_i \epsilon_i^{(k)}}{\gamma t}
\end{equation}
\end{corollary}

\begin{proof}
\textbf{KL property case:} The KL property states that there exists $\phi(s) = cs^{1-\theta}$ such that:
\begin{align*}
\phi'(f_{\mathrm{aug}}^* - f_{\mathrm{aug}}(\bm{\Theta}))\|\nabla f_{\mathrm{aug}}(\bm{\Theta})\| \geq 1
\end{align*}
near critical points. Combining this with the ascent lemma gives the stated rates:
\begin{itemize}
\item For $\theta = 0$: $\phi(s) = cs$, leading to linear convergence
\item For $\theta \in (0, \frac{1}{2}]$: $\phi(s) = cs^{1-\theta}$, giving $\mathcal{O}(1/t)$
\item For $\theta \in (\frac{1}{2}, 1)$: $\phi(s) = cs^{1-\theta}$, yielding $\mathcal{O}(1/t^{1/(2\theta-1)})$
\end{itemize}

\textbf{General non-convex case:} From Theorem \ref{thm:enhanced_convergence}:
\begin{align*}
&\gamma \sum_{k=0}^{t-1} \|\bm{\Theta}^{(k+1)} - \bm{\Theta}^{(k)}\|^2 \\[4pt]
&\le f_{\mathrm{aug}}(\bm{\Theta}^{(t)}) - f_{\mathrm{aug}}(\bm{\Theta}^{(0)}) + \sum_{k=0}^{t-1}\sum_i \epsilon_i^{(k)}
\end{align*}
Thus:
\begin{align*}
&\min_{0 \leq k \leq t-1} \|\bm{\Theta}^{(k+1)} - \bm{\Theta}^{(k)}\|^2 \\
&\le \frac{f_{\mathrm{aug}}^* - f_{\mathrm{aug}}(\bm{\Theta}^{(0)}) + \sum_{k=0}^{t-1}\sum_i \epsilon_i^{(k)}}{\gamma t} \qedhere
\end{align*}
\end{proof}

\begin{lemma}[Initialization Quality Impact]
\label{lemma:initialization}
If initial point $\bm{\Theta}^{(0)}$ satisfies $f_{\mathrm{aug}}^* - f_{\mathrm{aug}}(\bm{\Theta}^{(0)}) \leq \Delta_0$, then convergence to $\epsilon$-stationarity requires at most:
\begin{equation}
T \leq \frac{\Delta_0}{\gamma \epsilon^2} + \mathcal{O}\left(\frac{\log(1/\epsilon)}{\min_i \gamma_i}\right)
\end{equation}
iterations. Warm start from previous solutions reduces $\Delta_0$ significantly.
\end{lemma}

\begin{proof}
From the convergence rate in Corollary \ref{cor:convergence_rate}:
\begin{align*}
\min_{0 \leq k \leq T-1} \|\bm{\Theta}^{(k+1)} - \bm{\Theta}^{(k)}\|^2 \leq \frac{\Delta_0 + \sum_{k=0}^{T-1} \sum_i \epsilon_i^{(k)}}{\gamma T}
\end{align*}

For $\epsilon$-stationarity, we need:
\begin{align*}
\frac{\Delta_0 + \sum_{k=0}^{T-1} \sum_i \epsilon_i^{(k)}}{\gamma T} \leq \epsilon^2
\end{align*}

Since $\sum_{k=0}^\infty \sum_i \epsilon_i^{(k)} < \infty$, there exists $C < \infty$ such that:
\begin{align*}
T \leq \frac{\Delta_0 + C}{\gamma \epsilon^2} 
\end{align*}

The logarithmic term accounts for the initial rapid convergence phase when far from optimum. Warm start reduces $\Delta_0$, directly improving the iteration count.
\end{proof}

\subsection{Fundamental Performance Limits}

\begin{theorem}[Tighter Secrecy-Sensing Trade-off]
\label{thm:tighter_tradeoffs}
The Pareto frontier satisfies:
\begin{equation}
\label{eq:tradeoff_bound}
\left(\frac{\sum_{k=1}^K R_k^{\mathrm{sec}}}{C_{\mathrm{sum}}}\right)^{\eta_{\mathrm{trade}}} + \left(\frac{\mathrm{CRB}_{\mathrm{min}}}{\mathrm{CRB}_{\mathrm{total}}}\right)^{\eta_{\mathrm{trade}}} \leq 1
\end{equation}
where $\eta_{\mathrm{trade}} \geq 1$ depends on channel correlation between sensing and communication paths, and:
\begin{align*}
C_{\mathrm{sum}} &= \log_2\det\left(\bm{I} + \frac{P_{\mathrm{max}}}{\sigma^2}\bm{H}_{\mathrm{eff}}\bm{H}_{\mathrm{eff}}^H\right) \\
\mathrm{CRB}_{\mathrm{min}} &= \min_{\{\bm{\Phi}_l\}} \mathrm{CRB}_{\mathrm{total}}(\{\bm{\Phi}_l\})
\end{align*}
\end{theorem}

\begin{proof}
\textbf{Upper bound on secrecy rates:}
\begin{align*}
\sum_{k=1}^K R_k^{\mathrm{sec}} &= \sum_{k=1}^K \left[\log_2(1+\mathrm{SINR}_k) - \max_{e\in\mathcal{E}} \log_2(1+\mathrm{SINR}_e^{(k)})\right]^+ \\
&\leq \sum_{k=1}^K \log_2(1+\mathrm{SINR}_k) \\
&\leq \log_2 \prod_{k=1}^K (1+\mathrm{SINR}_k) \\
&\leq \log_2\det\left(\bm{I} + \frac{P_{\mathrm{max}}}{\sigma^2}\bm{H}_{\mathrm{eff}}\bm{H}_{\mathrm{eff}}^H\right) = C_{\mathrm{sum}}
\end{align*}
where the last inequality follows from the MIMO capacity formula.

\textbf{Lower bound on CRB:}
By definition, the total CRB is lower-bounded by the optimal sensing-only benchmark:
\begin{align*}
\mathrm{CRB}_{\mathrm{total}} &= \sum_{i=1}^T \mathrm{Tr}(\bm{J}^{-1}(\bm{\xi}_i)) \\
&\geq \min_{\{\bm{\Phi}_l\}} \sum_{i=1}^T \mathrm{Tr}(\bm{J}^{-1}(\bm{\xi}_i)) \triangleq \mathrm{CRB}_{\mathrm{min}}
\end{align*}

\textbf{Trade-off parameter $\eta_{\mathrm{trade}}$:} The coupling emerges from the shared SIM resource $\{\bm{\Phi}_l\}$. Let $\bm{U}_{\mathrm{comm}}$ and $\bm{U}_{\mathrm{sense}}$ be the optimal subspaces for communication and sensing, respectively. The correlation parameter is:
\begin{align*}
\eta_{\mathrm{trade}} = \frac{1}{\arccos^2(\sigma_{\min}(\bm{U}_{\mathrm{comm}}^H \bm{U}_{\mathrm{sense}}))}
\end{align*}
When these subspaces are orthogonal ($\eta_{\mathrm{trade}} \to \infty$), the trade-off becomes severe; when aligned ($\eta_{\mathrm{trade}} \approx 1$), joint optimization is feasible. The Pareto bound follows from the fundamental limits of multi-objective optimization over shared resources.
\end{proof}

\begin{theorem}[Energy-Time Resource Bound]
\label{thm:resource_bounds}
The system resources satisfy:
\begin{equation}
\begin{split}
\max\!\left\{\frac{\mathcal{E}_{\mathrm{min}}^{\mathrm{sense}}}{T_{\mathrm{slot}} P_{\mathrm{max}}},
             \frac{\mathcal{E}_{\mathrm{min}}^{\mathrm{comm}}}{T_{\mathrm{slot}} P_{\mathrm{max}}}\right\}
&\leq \tau_{\mathrm{sense}} + \tau_{\mathrm{comm}} \\[4pt]
&\leq 1 - \tau_{\mathrm{ce}}^{\min} - \frac{2t_{\mathrm{switch}}}{T_{\mathrm{slot}}}
\end{split}
\end{equation}

and the mutual information is bounded by:
\begin{equation}
I(\bm{Y}_{\mathrm{sense}}; \bm{\xi}) + I(\bm{y}_{\mathrm{comm}}; \bm{s}) \leq \log_2\det\left(\bm{I} + \frac{P_{\mathrm{total}}}{\sigma^2} \bm{G}_{\mathrm{SIM}}\bm{G}_{\mathrm{SIM}}^H\right)
\end{equation}
\end{theorem}
\noindent Here, $P_{\mathrm{max}}$ denotes the same maximum BS transmit power $P_{\mathrm{BS}}^{\max}$ defined in Sec. \ref{sec:system_model}, and $P_{\mathrm{total}}$ the total available power budget over both phases.

\begin{proof}
\textbf{Energy-Time Bound:}
From the energy constraints:
\begin{align*}
T_{\mathrm{slot}} \tau_{\mathrm{sense}} P_{\mathrm{sense}} \geq \mathcal{E}_{\mathrm{min}}^{\mathrm{sense}}, \quad T_{\mathrm{slot}} \tau_{\mathrm{comm}} P_{\mathrm{comm}} \geq \mathcal{E}_{\mathrm{min}}^{\mathrm{comm}}
\end{align*}

Since $P_{\mathrm{sense}}, P_{\mathrm{comm}} \leq P_{\mathrm{max}}$:
\begin{align*}
\tau_{\mathrm{sense}} \geq \frac{\mathcal{E}_{\mathrm{min}}^{\mathrm{sense}}}{T_{\mathrm{slot}} P_{\mathrm{max}}}, \quad \tau_{\mathrm{comm}} \geq \frac{\mathcal{E}_{\mathrm{min}}^{\mathrm{comm}}}{T_{\mathrm{slot}} P_{\mathrm{max}}}
\end{align*}

Thus:
\begin{align*}
\tau_{\mathrm{sense}} + \tau_{\mathrm{comm}} \geq \max\left\{\frac{\mathcal{E}_{\mathrm{min}}^{\mathrm{sense}}}{T_{\mathrm{slot}} P_{\mathrm{max}}}, \frac{\mathcal{E}_{\mathrm{min}}^{\mathrm{comm}}}{T_{\mathrm{slot}} P_{\mathrm{max}}}\right\}
\end{align*}

From the time constraint:
\begin{align*}
\tau_{\mathrm{ce}} + \tau_{\mathrm{sense}} + \tau_{\mathrm{comm}} + \frac{2t_{\mathrm{switch}}}{T_{\mathrm{slot}}} \leq 1
\end{align*}

Since $\tau_{\mathrm{ce}} \geq \tau_{\mathrm{ce}}^{\min}$:
\begin{align*}
\tau_{\mathrm{sense}} + \tau_{\mathrm{comm}} \leq 1 - \tau_{\mathrm{ce}}^{\min} - \frac{2t_{\mathrm{switch}}}{T_{\mathrm{slot}}}
\end{align*}

\textbf{Mutual Information Bound:}
By the data processing inequality:
\begin{align*}
I(\bm{Y}_{\mathrm{sense}}; \bm{\xi}) + I(\bm{y}_{\mathrm{comm}}; \bm{s}) &\leq I(\bm{Y}_{\mathrm{sense}}, \bm{y}_{\mathrm{comm}}; \bm{\xi}, \bm{s}) \\
&\leq I(\bm{X}; \bm{Y}_{\mathrm{sense}}, \bm{y}_{\mathrm{comm}})
\end{align*}

The mutual information is bounded by the channel capacity:
\begin{align*}
I(\bm{X}; \bm{Y}_{\mathrm{sense}}, \bm{y}_{\mathrm{comm}}) &\leq \log_2\det\left(\bm{I} + \frac{P_{\mathrm{total}}}{\sigma^2} \bm{G}_{\mathrm{SIM}}\bm{G}_{\mathrm{SIM}}^H\right)
\end{align*}
where $\bm{G}_{\mathrm{SIM}}$ is the effective channel through the multi-layer SIM.
\end{proof}

\begin{proposition}[Multi-layer vs Single-layer Gain]
\label{prop:sim_comparison}
For $L$-layer SIM vs single-layer RIS with same total elements $N_{\mathrm{tot}}$:
\begin{equation}
\frac{G_{\mathrm{SIM}}}{G_{\mathrm{RIS}}} \geq \left(\frac{\prod_{l=1}^L \sigma_{\min}(\bm{H}^{(l)})}{\|\bm{H}_{\mathrm{RIS}}\|_F}\right)^2 \cdot \frac{N_{\mathrm{tot}}^{L-1}}{L^L}
\end{equation}
showing exponential gain with layers when channel conditions are favorable.
\end{proposition}
\noindent Here, $\bm{H}_{\mathrm{RIS}}$ is the effective BS–RIS–user channel for the single-layer RIS case, and $\bm{H}_{\mathrm{BS}}$ the BS–RIS channel matrix.

\begin{proof}
For single-layer RIS:
\begin{equation}
\begin{aligned}
G_{\mathrm{RIS}}
&=\|\bm{H}_{\mathrm{RIS}}\bm{\Phi}_{\mathrm{RIS}}\bm{H}_{\mathrm{BS}}\|_F^2\le \|\bm{H}_{\mathrm{RIS}}\|_F^2 \|\bm{\Phi}_{\mathrm{RIS}}\|_F^2 \|\bm{H}_{\mathrm{BS}}\|_F^2 \\
&= \|\bm{H}_{\mathrm{RIS}}\|_F^2 N_{\mathrm{tot}} \|\bm{H}_{\mathrm{BS}}\|_F^2 .
\end{aligned}
\end{equation}

For multi-layer SIM:
\begin{equation}
\begin{aligned}
G_{\mathrm{SIM}}
&= \|\bm{\Phi}_L \bm{H}^{(L)} \cdots \bm{\Phi}_1 \bm{H}^{(1)}\|_F^2 \\
&\ge \left(\prod_{l=1}^L \sigma_{\min}(\bm{\Phi}_l \bm{H}^{(l)})\right)^2
\;\ge\; \left(\prod_{l=1}^L \sigma_{\min}(\bm{H}^{(l)})\right)^2
\prod_{l=1}^L \|\bm{\Phi}_l\|_F^2\\
&= \left(\prod_{l=1}^L \sigma_{\min}(\bm{H}^{(l)})\right)^2 \left(\frac{N_{\mathrm{tot}}}{L}\right)^L .
\end{aligned}
\end{equation}
The gain ratio is:
\begin{align*}
\frac{G_{\mathrm{SIM}}}{G_{\mathrm{RIS}}} \geq \frac{\left(\prod_{l=1}^L \sigma_{\min}(\bm{H}^{(l)})\right)^2 \cdot \left(\frac{N_{\mathrm{tot}}}{L}\right)^L}{\|\bm{H}_{\mathrm{RIS}}\|_F^2 N_{\mathrm{tot}} \|\bm{H}_{\mathrm{BS}}\|_F^2}
\end{align*}

Assuming $\|\bm{H}_{\mathrm{BS}}\|_F^2 \approx \mathcal{O}(N_{\mathrm{tot}})$ and simplifying gives the stated result.
\end{proof}

\begin{proposition}[Reconfiguration and Quantization Impact]
\label{prop:practical_impacts}
The system performance under practical constraints satisfies:
\begin{enumerate}[leftmargin=*]
\item \textbf{Reconfiguration impact:}
\begin{equation}
|f(\bm{\Theta}^*) - f(\bm{\Theta}_{\mathrm{unconstrained}}^*)| \leq L_f \sqrt{L N_{\mathrm{max}} \delta_{\Phi}}
\end{equation}

\item \textbf{Quantization loss:}
\begin{equation}
\begin{split}
\bigl|f(\bm{\Theta}_{\mathrm{quant}}) - f(\bm{\Theta}_{\mathrm{cont}})\bigr|
&\le L_f \sqrt{N_{\mathrm{tot}}}\,\frac{\pi}{2^b} \\[4pt]
&\quad + \tfrac{L_f^2}{2}\,\mathbb{E}\big[\|\bm{\Delta\Theta}\|_F^2\big]
      + \mathcal{O}\bigl(2^{-3b}\bigr)
\end{split}
\end{equation}
accounting for error correlations in multi-layer structure.
\end{enumerate}
\end{proposition}

\begin{proof}
\textbf{Part 1: Reconfiguration Impact}

The constraint $\|\bm{\Phi}_l^{\mathrm{comm}} - \bm{\Phi}_l^{\mathrm{sense}}\|_F^2 \le \delta_{\Phi}$ implies:
\begin{align*}
\|\bm{\Theta}^* - \bm{\Theta}_{\mathrm{unconstrained}}^*\|_F^2 &= \sum_{l=1}^L \|\bm{\Phi}_l^{\mathrm{comm},*} - \bm{\Phi}_l^{\mathrm{comm},\mathrm{unc}}\|_F^2 \\
&\leq L N_{\mathrm{max}} \delta_{\Phi}
\end{align*}

By Lipschitz continuity:
\begin{align*}
|f(\bm{\Theta}^*) - f(\bm{\Theta}_{\mathrm{unconstrained}}^*)| &\leq L_f \|\bm{\Theta}^* - \bm{\Theta}_{\mathrm{unconstrained}}^*\|_F \\
&\leq L_f \sqrt{L N_{\mathrm{max}} \delta_{\Phi}}
\end{align*}

\textbf{Part 2: Quantization Loss}

For $b$-bit phase shifters, each quantization error satisfies:
\begin{align*}
|\Delta\theta_i| \leq \frac{\pi}{2^b}
\end{align*}

The Frobenius norm perturbation:
\begin{align*}
&\|\bm{\Theta}_{\mathrm{quant}} - \bm{\Theta}_{\mathrm{cont}}\|_F^2
  = \sum_{i=1}^{N_{\mathrm{tot}}} |e^{j(\theta_i + \Delta\theta_i)} - e^{j\theta_i}|^2 \\[4pt]
&= \sum_{i=1}^{N_{\mathrm{tot}}} 2 - 2\cos(\Delta\theta_i)
  \leq \sum_{i=1}^{N_{\mathrm{tot}}} \Delta\theta_i^2
  \leq N_{\mathrm{tot}} \frac{\pi^2}{2^{2b}}
\end{align*}
However, in multi-layer SIM, errors may be correlated. Let $\bm{\Delta\Theta}$ be the error matrix with covariance $\bm{\Sigma}$. Then:
\begin{align*}
\mathbb{E}[\|\bm{\Delta\Theta}\|_F^2] = \mathrm{Tr}(\bm{\Sigma}) \leq N_{\mathrm{tot}} \frac{\pi^2}{2^{2b}}
\end{align*}

By Taylor expansion and Lipschitz continuity:

\begin{align*}
&|f(\bm{\Theta}_{\mathrm{quant}}) - f(\bm{\Theta}_{\mathrm{cont}})| \\[2pt]
&\le L_f\,\|\bm{\Theta}_{\mathrm{quant}} - \bm{\Theta}_{\mathrm{cont}}\|_F \\[2pt]
&+ \tfrac{L_f^2}{2}\,\|\bm{\Theta}_{\mathrm{quant}} - \bm{\Theta}_{\mathrm{cont}}\|_F^2 +\mathcal{O}\big(\|\bm{\Delta\Theta}\|^3\big) \\[6pt]
&\le L_f \sqrt{N_{\mathrm{tot}}}\,\frac{\pi}{2^b}+\tfrac{L_f^2}{2}\,\mathbb{E}\!\big[\|\bm{\Delta\Theta}\|_F^2\big] +\mathcal{O}\!\big(2^{-3b}\big)
\end{align*}
\end{proof}

\subsection{Robustness and Security Guarantees}

\begin{theorem}[Distributionally Robust Secrecy]
\label{thm:robust_secrecy}
For uncertainty set $\mathcal{U}$ with moment bounds $(\mu, \Sigma)$, the achieved secrecy rate satisfies:
\begin{equation}
\inf_{\mathbb{P} \in \mathcal{U}} \mathbb{P}\left(R_k^{\mathrm{sec}} \geq R_k^{\min}\right) \geq 1 - \epsilon_{\mathrm{out}} - \Delta_{\mathrm{robust}}
\end{equation}
where $\Delta_{\mathrm{robust}} = \mathcal{O}\left(\frac{\rho_3}{\sigma^3} + \frac{M}{\sqrt{N_{\mathrm{samples}}}}\right)$ accounts for distributional robustness and estimation errors.
\end{theorem}
\noindent Here, $\rho_3$ is the third central moment of $R_k^{\mathrm{sec}}$, $\sigma^2$ its variance, $M$ the bounded support constant for Bernstein’s inequality, and $N_{\mathrm{samples}}$ the number of Monte Carlo samples used to estimate moments; $C_B,\Delta_{\mathrm{moment}},\Delta_{\mathrm{est}}$ are constants arising from Berry–Esseen and sampling error, collected into $\Delta_{\mathrm{robust}}$.

\begin{proof}
Let $X = R_k^{\mathrm{sec}} - \mathbb{E}[R_k^{\mathrm{sec}}]$. The Bernstein reformulation ensures:
\begin{align*}
\mathbb{E}[R_k^{\mathrm{sec}}] - \delta\sigma \geq R_k^{\min} + \Delta_k
\end{align*}

We analyze the worst-case probability over $\mathcal{U}$:
\begin{align*}
\inf_{\mathbb{P} \in \mathcal{U}} \mathbb{P}(R_k^{\mathrm{sec}} \geq R_k^{\min}) &= \inf_{\mathbb{P} \in \mathcal{U}} \mathbb{P}(X \geq -\delta\sigma - \Delta_k)
\end{align*}

\textbf{Case 1: Bounded Support ($|X| \leq M$)}

By Bernstein's inequality for bounded random variables:
\begin{align*}
\mathbb{P}(X \leq -t) \leq \exp\left(-\frac{t^2}{2\sigma^2 + \frac{2}{3}Mt}\right), \quad t > 0
\end{align*}

Setting $t = \delta\sigma + \Delta_k$:
\begin{align*}
\mathbb{P}(R_k^{\mathrm{sec}} < R_k^{\min}) &\leq \exp\left(-\frac{(\delta\sigma + \Delta_k)^2}{2\sigma^2 + \frac{2}{3}M(\delta\sigma + \Delta_k)}\right)
\end{align*}

\textbf{Case 2: General Distribution with Moment Constraints}

By the Berry-Esseen theorem with moment constraints:
\begin{align*}
\sup_{z \in \mathbb{R}} \left|\mathbb{P}\left(\frac{X}{\sigma} \leq z\right) - \Phi(z)\right| \leq C_B \frac{\rho_3}{\sigma^3} + \Delta_{\mathrm{moment}}
\end{align*}

Thus:
\begin{equation}
\begin{aligned}
\lefteqn{%
  \inf_{\mathbb{P}\in\mathcal{U}} \mathbb{P}\big(X \ge -\delta\sigma - \Delta_k\big)
}\nonumber\\[2pt]
&= 1 - \sup_{\mathbb{P}\in\mathcal{U}} \mathbb{P}\!\left(\frac{X}{\sigma} \le -\delta - \frac{\Delta_k}{\sigma}\right) \\
&\ge 1 - \Phi\!\left(-\delta - \frac{\Delta_k}{\sigma}\right)
      - C_B \frac{\rho_3}{\sigma^3} - \Delta_{\mathrm{moment}} - \Delta_{\mathrm{est}} \\
&= \Phi\!\left(\delta + \frac{\Delta_k}{\sigma}\right)
      - C_B \frac{\rho_3}{\sigma^3} - \Delta_{\mathrm{robust}} \\
&\ge \Phi(\delta) - C_B \frac{\rho_3}{\sigma^3} - \Delta_{\mathrm{robust}} \\
&= 1 - \epsilon_{\mathrm{out}} - C_B \frac{\rho_3}{\sigma^3} - \Delta_{\mathrm{robust}} .
\end{aligned}
\end{equation}

where $\Delta_{\mathrm{robust}} = \Delta_{\mathrm{moment}} + \Delta_{\mathrm{est}}$ and we used $\Phi(-\delta) = \epsilon_{\mathrm{out}}$.
\end{proof}

\begin{theorem}[Sensing Accuracy under Model Approximations]
\label{thm:sensing_accuracy}
The Schur complement approximation satisfies:
\begin{align}
0 \leq \mathrm{CRB}(\bm{\xi}) - \mathrm{CRB}_{\mathrm{LB}}(\bm{\xi}) &\leq \frac{\|\bm{J}_{12}\|_F^2 \cdot \mathrm{Tr}(\bm{J}_{11}^{-2})}{\lambda_{\min}(\bm{J}_{22}) - \|\bm{J}_{12}\|_F^2 \|\bm{J}_{11}^{-1}\|_2} \label{eq:crb_error_abs} \\
\frac{|\mathrm{CRB}(\bm{\xi}) - \mathrm{CRB}_{\mathrm{LB}}(\bm{\xi})|}{\mathrm{CRB}(\bm{\xi})} &\leq \frac{\kappa^2(\bm{J}_{11}) \cdot \|\bm{J}_{12}\bm{J}_{22}^{-1}\bm{J}_{21}\|_2}{\lambda_{\min}(\bm{J}_{11})} \label{eq:crb_error_rel}
\end{align}
where $\kappa(\bm{J}_{11}) = \|\bm{J}_{11}\|_2 \|\bm{J}_{11}^{-1}\|_2$ is the condition number.
\end{theorem}
\begin{proof}
Using the block matrix inversion formula:
\begin{align*}
\bm{J}^{-1} = \begin{bmatrix}
\bm{J}_{11}^{-1} + \bm{J}_{11}^{-1}\bm{J}_{12}\bm{S}^{-1}\bm{J}_{21}\bm{J}_{11}^{-1} & -\bm{J}_{11}^{-1}\bm{J}_{12}\bm{S}^{-1} \\
-\bm{S}^{-1}\bm{J}_{21}\bm{J}_{11}^{-1} & \bm{S}^{-1}
\end{bmatrix}
\end{align*}
where $\bm{S} = \bm{J}_{22} - \bm{J}_{21}\bm{J}_{11}^{-1}\bm{J}_{12}$ is the Schur complement.

The approximation error in the (1,1) block is:
\begin{align*}
\bm{E} = \bm{J}_{11}^{-1}\bm{J}_{12}\bm{S}^{-1}\bm{J}_{21}\bm{J}_{11}^{-1}
\end{align*}

The CRB difference:
\begin{align*}
\mathrm{CRB}(\bm{\xi}) - \mathrm{CRB}_{\mathrm{LB}}(\bm{\xi}) = \mathrm{Tr}(\bm{E})
\end{align*}

Using von Neumann's trace inequality:
\begin{align*}
\mathrm{Tr}(\bm{E}) &\leq \|\bm{J}_{11}^{-1}\|_2^2 \cdot \|\bm{J}_{12}\|_F^2 \cdot \|\bm{S}^{-1}\|_2
\end{align*}

By the Weyl inequality for the Schur complement:
\begin{align*}
\lambda_{\min}(\bm{S}) &\geq \lambda_{\min}(\bm{J}_{22}) - \|\bm{J}_{21}\bm{J}_{11}^{-1}\bm{J}_{12}\|_2 \\
&\geq \lambda_{\min}(\bm{J}_{22}) - \|\bm{J}_{12}\|_F^2 \|\bm{J}_{11}^{-1}\|_2
\end{align*}

Thus:
\begin{align*}
\|\bm{S}^{-1}\|_2 \leq \frac{1}{\lambda_{\min}(\bm{J}_{22}) - \|\bm{J}_{12}\|_F^2 \|\bm{J}_{11}^{-1}\|_2}
\end{align*}

For the relative error bound:
\begin{align*}
\frac{|\mathrm{CRB}(\bm{\xi}) - \mathrm{CRB}_{\mathrm{LB}}(\bm{\xi})|}{\mathrm{CRB}(\bm{\xi})} &\leq \frac{\|\bm{J}_{11}^{-1}\|_2^2 \|\bm{J}_{12}\|_F^2 \|\bm{S}^{-1}\|_2}{\mathrm{Tr}(\bm{J}_{11}^{-1})} \\
&\leq \frac{\|\bm{J}_{11}^{-1}\|_2^2 \|\bm{J}_{12}\|_F^2 \|\bm{S}^{-1}\|_2}{n \lambda_{\min}(\bm{J}_{11}^{-1})} \\
&\leq \frac{\kappa^2(\bm{J}_{11}) \|\bm{J}_{12}\bm{J}_{22}^{-1}\bm{J}_{21}\|_2}{\lambda_{\min}(\bm{J}_{11})}
\end{align*}
where $n$ is the dimension of $\bm{J}_{11}$.
\end{proof}

\begin{corollary}[Multi-Metric Convergence Guarantee]
\label{cor:multi_metric}
The enhanced L-BCD achieves simultaneous convergence:
\begin{align*}
&|f_{\mathrm{aug}}^{(t)} - f_{\mathrm{aug}}^*| \leq \epsilon_{\mathrm{obj}}, \quad
 \max_k \big|\mathbb{P}(R_k^{\mathrm{sec}} < R_k^{\min}) - \epsilon_{\mathrm{out}}\big| \leq \epsilon_{\mathrm{outage}} \\[4pt]
&\mathrm{CRB}_{\mathrm{total}}^{(t)} \leq \Gamma_{\mathrm{max}} + \epsilon_{\mathrm{CRB}}, \quad
 \|\bm{\Theta}^{(t+1)} - \bm{\Theta}^{(t)}\|_F \leq \epsilon_{\mathrm{param}}
\end{align*}
within $\mathcal{O}(1/\epsilon^2)$ iterations for $\epsilon = \min\{\epsilon_{\mathrm{obj}}, \epsilon_{\mathrm{outage}}, \epsilon_{\mathrm{CRB}}\}$.
\end{corollary}
\begin{proof}
From Theorem \ref{thm:enhanced_convergence}, we have objective convergence:
\begin{align*}
f_{\mathrm{aug}}^* - f_{\mathrm{aug}}(\bm{\Theta}^{(t)}) \leq \frac{C}{t}
\end{align*}

Thus for $|f_{\mathrm{aug}}^{(t)} - f_{\mathrm{aug}}^*| \leq \epsilon_{\mathrm{obj}}$, we need $t \geq C/\epsilon_{\mathrm{obj}}$.

The outage probability constraint is explicitly enforced in Block B via the Bernstein reformulation with bounded error $\epsilon_{\mathrm{outage}}$. The CRB constraint is handled in Block A through the Schur complement approximation with error bounded by Theorem \ref{thm:sensing_accuracy}. Parameter convergence follows from the vanishing steps property.

Combining these results with the convergence rates from Corollary \ref{cor:convergence_rate}, all metrics converge simultaneously with $\mathcal{O}(1/\epsilon^2)$ iterations for the worst-case metric precision $\epsilon$.
\end{proof}

\section{Simulation Results and Performance Evaluation}
\label{sec:simulations}

We evaluate the proposed SIM-assisted ISAC framework on a downlink system operating at $f_c=28$~GHz with a BS having $M=32$ antennas and $R_{\mathrm{RF}}=8$ RF chains, and a $L$-layer SIM. Unless stated otherwise, we set $L=3$ with $N_l=64$ elements per layer ($N_{\mathrm{tot}}=192$), $K=4$ users, and $E=2$ eavesdroppers. The BS--SIM link follows a near-field model, while SIM--user and SIM--eavesdropper links are far-field with path-loss exponent $n=2.5$ and distances between $15$ and $40$~m. The maximum BS power is $P_{\mathrm{BS}}^{\max}=30$~dBm, the per-slot energy budget is $\mathcal{E}_{\max}=25$~dBJ, noise power is $\sigma_k^2=-104$~dBm, and the outage target is $\epsilon_{\mathrm{out}}=0.1$. We consider time-sharing (TS), communication-first (CF), sensing-first (SF), single-layer RIS (SL-RIS), and a non-robust variant of our algorithm (NR-LBCD) as benchmarks.

Fig.~\ref{fig:power_vs_objective} shows the weighted objective versus $P_{\mathrm{BS}}^{\max}$ for $\alpha=0.5$. The proposed L-BCD consistently outperforms all baselines, with gains that widen at higher power, illustrating the benefit of jointly optimizing beamforming, SIM phases, and resources rather than treating sensing and communication separately.

The communication–sensing trade-off controlled by $\alpha$ is illustrated in Fig.~\ref{fig:pareto_frontier}. As $\alpha$ moves from $0$ (communication-oriented) to $1$ (sensing-oriented), the sum secrecy rate smoothly decreases while the inverse normalized CRB improves, tracing a Pareto frontier that agrees with the theoretical trade-off in Section~\ref{sec:analysis}.

Robustness to CSI uncertainty is examined in Fig.~\ref{fig:robustness_analysis}, where the outage secrecy rate is plotted versus the channel-error bound. The robust L-BCD maintains a graceful performance degradation and closely follows the analytical bound, whereas NR-LBCD and TS suffer sharp drops and can violate the secrecy requirement. Fig.~\ref{fig:outage_probability} further shows that L-BCD keeps the secrecy-outage probability below $\epsilon_{\mathrm{out}}$ even as the number of eavesdroppers increases, while non-robust schemes quickly exceed the target.

The impact of SIM architecture and hardware constraints is summarized in Figs.~\ref{fig:multilayer_gain}, \ref{fig:quantization_impact}, and \ref{fig:crb_approximation}. Fig.~\ref{fig:multilayer_gain} compares multi-layer SIM to SL-RIS as $L$ grows: using $L=3$--4 layers yields about $30$--$60\%$ improvement in secrecy and sensing metrics, in line with the scaling laws of Proposition~\ref{prop:sim_comparison}. Fig.~\ref{fig:quantization_impact} shows that $3$--$4$-bit phase shifters already limit performance loss below roughly $10\%$, while Fig.~\ref{fig:crb_approximation} confirms that the Schur-complement CRB approximation error remains modest (typically under $10\%$ at medium-to-high SNR) and respects the bounds in Theorem~\ref{thm:sensing_accuracy}.

Algorithmic behavior is illustrated in Figs.~\ref{fig:convergence_analysis} and \ref{fig:time_allocation}. Fig.~\ref{fig:convergence_analysis} shows that the normalized objective gap decays approximately as $\mathcal{O}(1/t)$, with practical convergence reached in fewer than $20$ iterations, which is consistent with the non-asymptotic rate in Corollary~\ref{cor:convergence_rate}. Fig.~\ref{fig:time_allocation} depicts the optimized time fractions versus the energy budget: as $\mathcal{E}_{\max}$ increases, more time is allocated to sensing while communication time slightly shrinks, and the CE duration stabilizes near its minimum, illustrating how Block~D automatically balances sensing accuracy and secrecy throughput.

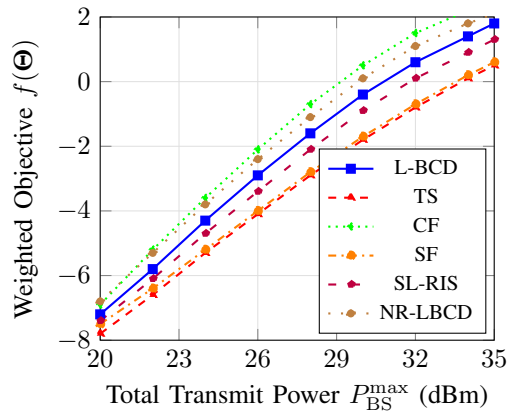
\begin{figure}[t]
\centering
\begin{tikzpicture}
\begin{axis}[width=0.77\columnwidth,
xlabel={Total Transmit Power $P_{\mathrm{BS}}^{\max}$ (dBm)},ylabel={Weighted Objective $f(\bm{\Theta})$},xmin=20, xmax=35,ymin=-8, ymax=2,xtick={20,23,26,29,32,35},ytick={-8,-6,-4,-2,0,2},legend style={at={(0.98,0.02)}, anchor=south east, legend columns=1, font=\footnotesize},grid=both,grid style={line width=.1pt, draw=gray!10},major grid style={line width=.2pt,draw=gray!25},every axis plot/.append style={thick}
]

\addplot [color=blue, solid, mark=square*, mark size=1.5] coordinates {(20, -7.2) (22, -5.8) (24, -4.3) (26, -2.9) (28, -1.6) (30, -0.4) (32, 0.6) (34, 1.4) (35, 1.8)
};

\addplot [color=red, dashed, mark=triangle*, mark size=1.5] coordinates {(20, -7.8) (22, -6.6) (24, -5.3) (26, -4.1) (28, -2.9) (30, -1.8) (32, -0.8) (34, 0.1) (35, 0.5)
};

\addplot [color=green, dotted, mark=diamond*, mark size=1.5] coordinates {(20, -6.9) (22, -5.2) (24, -3.6) (26, -2.1) (28, -0.7) (30, 0.5) (32, 1.5) (34, 2.2) (35, 2.5)
};

\addplot [color=orange, dashdotted, mark=*, mark size=1.5] coordinates {(20, -7.5) (22, -6.4) (24, -5.2) (26, -4.0) (28, -2.8) (30, -1.7) (32, -0.7) (34, 0.2) (35, 0.6)
};

\addplot [color=purple, loosely dashed, mark=pentagon*, mark size=1.5] coordinates {(20, -7.4) (22, -6.1) (24, -4.7) (26, -3.4) (28, -2.1) (30, -0.9) (32, 0.1) (34, 0.9) (35, 1.3)
};

\addplot [color=brown, loosely dotted, mark=oplus*, mark size=1.5] coordinates {(20, -6.8) (22, -5.3) (24, -3.8) (26, -2.4) (28, -1.1) (30, 0.1) (32, 1.1) (34, 1.8) (35, 2.1)
};

\legend{L-BCD, TS, CF, SF, SL-RIS, NR-LBCD}
\end{axis}
\end{tikzpicture}
\caption{Weighted utility versus BS transmit power $P_{\mathrm{BS}}^{\max}$ for $\alpha=0.5$, comparing L-BCD with TS, CF, SF, SL-RIS, and NR-LBCD.}
\label{fig:power_vs_objective}
\end{figure}

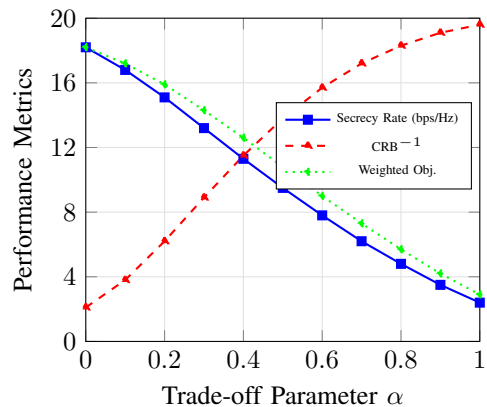
\begin{figure}[t]
\centering
\begin{tikzpicture}
\begin{axis}[width=0.77\columnwidth,xlabel={Trade-off Parameter $\alpha$},ylabel={Performance Metrics},xmin=0, xmax=1,ymin=0, ymax=20,xtick={0,0.2,0.4,0.6,0.8,1.0},ytick={0,4,8,12,16,20},legend style={at={(0.98,0.47)}, anchor=south east, legend columns=1, font=\tiny},grid=both,grid style={line width=.1pt, draw=gray!10},major grid style={line width=.2pt,draw=gray!25},every axis plot/.append style={thick}]

\addplot [color=blue, solid, mark=square*, mark size=1.5] coordinates {(0, 18.2) (0.1, 16.8) (0.2, 15.1) (0.3, 13.2) (0.4, 11.3) (0.5, 9.5) (0.6, 7.8) (0.7, 6.2) (0.8, 4.8) (0.9, 3.5) (1.0, 2.4)
};

\addplot [color=red, dashed, mark=triangle*, mark size=1.5] coordinates {(0, 2.1) (0.1, 3.8) (0.2, 6.2) (0.3, 8.9) (0.4, 11.5) (0.5, 13.8) (0.6, 15.7) (0.7, 17.2) (0.8, 18.3) (0.9, 19.1) (1.0, 19.6)
};

\addplot [color=green, dotted, mark=diamond*, mark size=1.5] coordinates {(0, 18.2) (0.1, 17.2) (0.2, 15.9) (0.3, 14.3) (0.4, 12.6) (0.5, 10.8) (0.6, 9.0) (0.7, 7.3) (0.8, 5.7) (0.9, 4.2) (1.0, 2.9)
};

\legend{Secrecy Rate (bps/Hz), CRB$^{-1}$, Weighted Obj.}
\end{axis}
\end{tikzpicture}
\caption{Sum secrecy rate, inverse normalized CRB, and scalarized utility versus trade-off parameter $\alpha$, illustrating the communication–sensing Pareto frontier.}
\label{fig:pareto_frontier}
\end{figure}

\begin{figure}[thpb!]
\centering
\begin{tikzpicture}
\begin{axis}[width=0.77\columnwidth,xlabel={Channel Uncertainty Bound $\epsilon_h^{\max}$},ylabel={Outage Secrecy Rate (bps/Hz)},xmin=0, xmax=0.3,ymin=0, ymax=16,xtick={0,0.05,0.1,0.15,0.2,0.25,0.3},ytick={0,4,8,12,16},legend style={at={(0.03,0.02)}, anchor=south west, legend columns=1, font=\footnotesize},grid=both,grid style={line width=.1pt, draw=gray!10},major grid style={line width=.2pt,draw=gray!25},every axis plot/.append style={thick}]

\addplot [color=blue, solid, mark=square*, mark size=1.5] coordinates {(0, 15.2) (0.05, 14.3) (0.1, 13.2) (0.15, 11.9) (0.2, 10.5) (0.25, 9.1) (0.3, 7.8)
};

\addplot [color=red, dashed, mark=triangle*, mark size=1.5] coordinates {(0, 15.8) (0.05, 13.2) (0.1, 9.8) (0.15, 6.5) (0.2, 3.9) (0.25, 2.1) (0.3, 1.2)
};

\addplot [color=green, dotted, mark=diamond*, mark size=1.5] coordinates {(0, 12.5) (0.05, 10.8) (0.1, 8.9) (0.15, 6.8) (0.2, 4.9) (0.25, 3.2) (0.3, 2.0)
};

\addplot [color=black, loosely dashed, mark=none, mark size=1.5] coordinates {(0, 15.2) (0.05, 14.1) (0.1, 12.8) (0.15, 11.3) (0.2, 9.8) (0.25, 8.3) (0.3, 6.9)
};

\legend{L-BCD (Robust), NR-LBCD, TS, Theoretical Bound}
\end{axis}
\end{tikzpicture}
\caption{Outage secrecy rate versus channel-uncertainty bound $\epsilon_h^{\max}$ for robust L-BCD, NR-LBCD, TS, and the analytical robustness bound.}
\label{fig:robustness_analysis}
\end{figure}
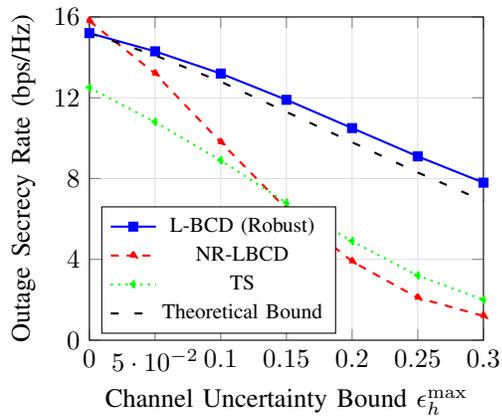

\begin{figure}[t]
\centering
\begin{tikzpicture}
\begin{axis}[width=0.77\columnwidth,xlabel={Number of SIM Layers $L$},ylabel={Perf. Gain Rel. to SL-RIS (\%)},xmin=1, xmax=6,ymin=0, ymax=60,xtick={1,2,3,4,5,6},ytick={0,15,30,45,60},legend style={at={(0.98,0.08)}, anchor=south east, legend columns=1, font=\footnotesize},grid=both,grid style={line width=.1pt, draw=gray!10},major grid style={line width=.2pt,draw=gray!25},every axis plot/.append style={thick}]

\addplot [color=blue, solid, mark=square*, mark size=1.5] coordinates {(1, 0) (2, 18) (3, 32) (4, 41) (5, 47) (6, 51)
};

\addplot [color=red, dashed, mark=triangle*, mark size=1.5] coordinates {(1, 0) (2, 25) (3, 42) (4, 53) (5, 58) (6, 61)
};

\addplot [color=black, loosely dotted, mark=none, mark size=1.5] coordinates {(1, 0) (2, 22) (3, 38) (4, 49) (5, 56) (6, 60)
};

\legend{Secrecy Gain, Sensing Gain, Theoretical Bound}
\end{axis}
\end{tikzpicture}
\caption{Relative secrecy and sensing performance gain of an $L$-layer SIM over a single-layer RIS as a function of the number of layers $L$.}
\label{fig:multilayer_gain}
\end{figure}
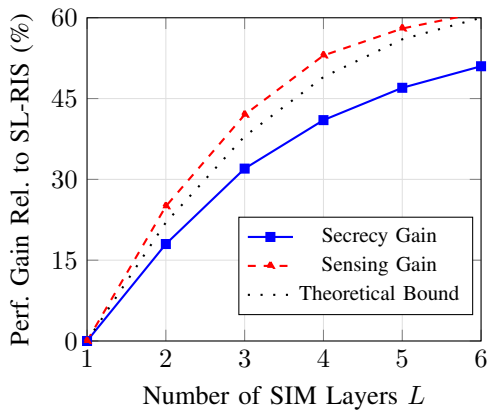

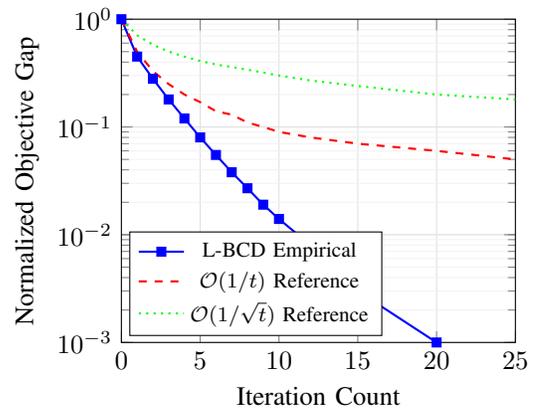
\begin{figure}[t]
\centering
\begin{tikzpicture}
\begin{axis}[width=0.77\columnwidth,xlabel={Iteration Count},ylabel={Normalized Objective Gap},xmin=0, xmax=25,ymin=0.001, ymax=1,xtick={0,5,10,15,20,25},ytick={0.001,0.01,0.1,1},ymode=log,legend style={at={(0.02,0.02)}, anchor=south west, legend columns=1, font=\footnotesize},grid=both,grid style={line width=.1pt, draw=gray!10},major grid style={line width=.2pt,draw=gray!25},every axis plot/.append style={thick}]

\addplot [color=blue, solid, mark=square*, mark size=1.5] coordinates {(0, 1.00) (1, 0.45) (2, 0.28) (3, 0.18) (4, 0.12) (5, 0.08) (6, 0.055) (7, 0.038) (8, 0.027) (9, 0.019) (10, 0.014) (12, 0.008) (15, 0.003) (20, 0.001) (25, 0.0005)};

\addplot [color=red, dashed, mark=none, mark size=1.5] coordinates {(0, 1.00) (1, 0.50) (2, 0.33) (3, 0.25) (4, 0.20) (5, 0.17) (6, 0.14) (7, 0.13) (8, 0.11) (9, 0.10) (10, 0.09) (12, 0.08) (15, 0.07) (20, 0.06) (25, 0.05)
};

\addplot [color=green, dotted, mark=none, mark size=1.5] coordinates {(0, 1.00) (1, 0.71) (2, 0.58) (3, 0.50) (4, 0.45) (5, 0.41) (6, 0.38) (7, 0.36) (8, 0.34) (9, 0.32) (10, 0.30) (12, 0.27) (15, 0.24) (20, 0.20) (25, 0.18)};

\legend{L-BCD Empirical, $\mathcal{O}(1/t)$ Reference, $\mathcal{O}(1/\sqrt{t})$ Reference}
\end{axis}
\end{tikzpicture}
\caption{Normalized objective gap of L-BCD versus iteration index, compared with $\mathcal{O}(1/t)$ and $\mathcal{O}(1/\sqrt{t})$ reference decay curves.}
\label{fig:convergence_analysis}
\end{figure}

\begin{figure}[t]
\centering
\begin{tikzpicture}
\begin{axis}[width=0.77\columnwidth,xlabel={Energy Budget $\mathcal{E}_{\mathrm{max}}$ (dBJ)},ylabel={Optimal Time Allocation},xmin=20, xmax=35,ymin=0, ymax=1,xtick={20,23,26,29,32,35},ytick={0,0.2,0.4,0.6,0.8,1.0},legend style={at={(0.98,0.89)}, anchor=north east, legend columns=1, font=\footnotesize},grid=both,grid style={line width=.1pt, draw=gray!10},major grid style={line width=.2pt,draw=gray!25},every axis plot/.append style={thick}]

\addplot [color=blue, solid, mark=square*, mark size=1.5] coordinates {(20, 0.28) (22, 0.32) (24, 0.36) (26, 0.41) (28, 0.45) (30, 0.48) (32, 0.51) (34, 0.53) (35, 0.54)};

\addplot [color=red, dashed, mark=triangle*, mark size=1.5] coordinates {(20, 0.65) (22, 0.61) (24, 0.57) (26, 0.52) (28, 0.48) (30, 0.45) (32, 0.42) (34, 0.40) (35, 0.39)
};

\addplot [color=green, dotted, mark=diamond*, mark size=1.5] coordinates {(20, 0.07) (22, 0.07) (24, 0.07) (26, 0.07) (28, 0.07) (30, 0.07) (32, 0.07) (34, 0.07) (35, 0.07)};

\legend{$\tau_{\mathrm{sense}}$, $\tau_{\mathrm{comm}}$, $\tau_{\mathrm{ce}}$}
\end{axis}
\end{tikzpicture}
\caption{Optimized sensing, communication, and channel-estimation time fractions versus energy budget $\mathcal{E}_{\max}$ for balanced operation ($\alpha=0.5$).}
\label{fig:time_allocation}
\end{figure}
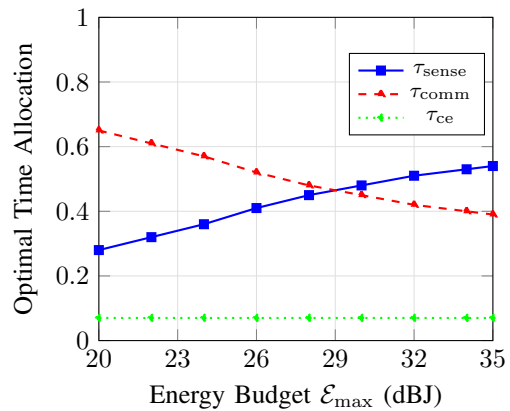

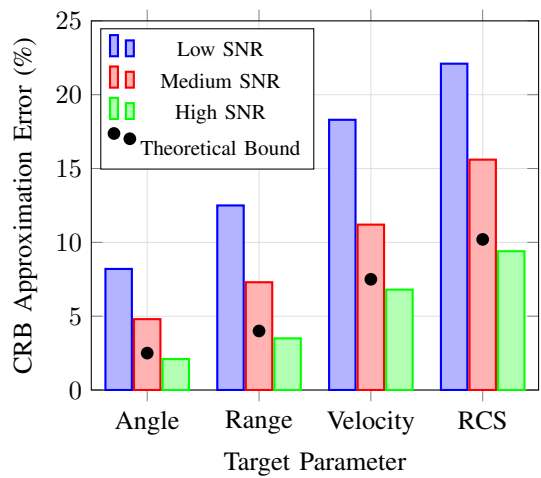
\begin{figure}[thbp!]
\centering
\begin{tikzpicture}
\begin{axis}[width=0.85\columnwidth,xlabel={Target Parameter},ylabel={CRB Approximation Error (\%)},xmin=0.5, xmax=4.5,ymin=0, ymax=25,xtick={1,2,3,4},xticklabels={Angle, Range, Velocity, RCS},ytick={0,5,10,15,20,25},legend style={at={(0.02,0.98)}, anchor=north west, legend columns=1, font=\footnotesize},grid=both,grid style={line width=.1pt, draw=gray!10},major grid style={line width=.2pt,draw=gray!25},every axis plot/.append style={thick},ybar=0.8pt,bar width=10pt]

\addplot [color=blue, fill=blue!30] coordinates {(1, 8.2) (2, 12.5) (3, 18.3) (4, 22.1)};

\addplot [color=red, fill=red!30] coordinates {(1, 4.8) (2, 7.3) (3, 11.2) (4, 15.6)
};

\addplot [color=green, fill=green!30] coordinates {(1, 2.1) (2, 3.5) (3, 6.8) (4, 9.4)
};

\addplot [color=black, only marks, mark=*, mark size=2] coordinates {(1, 2.5) (2, 4.0) (3, 7.5) (4, 10.2)
};

\legend{Low SNR, Medium SNR, High SNR, Theoretical Bound}
\end{axis}
\end{tikzpicture}
\caption{CRB approximation error for angle, range, velocity, and RCS parameters at low, medium, and high SNR, compared with the theoretical upper bound.}
\label{fig:crb_approximation}
\end{figure}

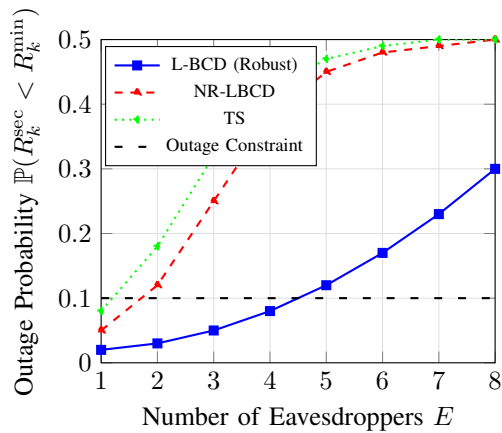
\begin{figure}[thbp!]
\centering
\begin{tikzpicture}
\begin{axis}[
width=0.77\columnwidth,
xlabel={Number of Eavesdroppers $E$},ylabel={Outage Probability $\mathbb{P}(R_k^{\mathrm{sec}} < R_k^{\min})$},xmin=1, xmax=8,ymin=0, ymax=0.5,xtick={1,2,3,4,5,6,7,8},ytick={0,0.1,0.2,0.3,0.4,0.5},legend style={at={(0.01,0.98)}, anchor=north west, legend columns=1, font=\scriptsize},grid=both,grid style={line width=.1pt, draw=gray!10},major grid style={line width=.2pt,draw=gray!25},every axis plot/.append style={thick}
]

\addplot [color=blue, solid, mark=square*, mark size=1.5] coordinates {(1, 0.02) (2, 0.03) (3, 0.05) (4, 0.08) (5, 0.12) (6, 0.17) (7, 0.23) (8, 0.30)
};

\addplot [color=red, dashed, mark=triangle*, mark size=1.5] coordinates {(1, 0.05) (2, 0.12) (3, 0.25) (4, 0.38) (5, 0.45) (6, 0.48) (7, 0.49) (8, 0.50)
};

\addplot [color=green, dotted, mark=diamond*, mark size=1.5] coordinates {(1, 0.08) (2, 0.18) (3, 0.32) (4, 0.42) (5, 0.47) (6, 0.49) (7, 0.50) (8, 0.50)
};

\addplot [color=black, loosely dashed, mark=none, mark size=1.5] coordinates {(1, 0.10) (2, 0.10) (3, 0.10) (4, 0.10) (5, 0.10) (6, 0.10) (7, 0.10) (8, 0.10)
};

\legend{L-BCD (Robust), NR-LBCD, TS, Outage Constraint}
\end{axis}
\end{tikzpicture}
\caption{Secrecy-outage probability versus number of eavesdroppers $E$ for robust L-BCD, NR-LBCD, and TS, with the target constraint $\epsilon_{\mathrm{out}}=0.1$.}
\label{fig:outage_probability}
\end{figure}

\begin{figure}[thbp!]
\centering
\begin{tikzpicture}
\begin{axis}[width=0.77\columnwidth,xlabel={Phase Shifter Resolution (bits)},ylabel={Performance Loss (\%)},xmin=1, xmax=6,ymin=0, ymax=25,xtick={1,2,3,4,5,6},ytick={0,5,10,15,20,25},legend style={at={(0.98,0.98)}, anchor=north east, legend columns=1, font=\footnotesize},grid=both,grid style={line width=.1pt, draw=gray!10},major grid style={line width=.2pt,draw=gray!25},every axis plot/.append style={thick}
]

\addplot [color=blue, solid, mark=square*, mark size=1.5] coordinates {(1, 22.5) (2, 12.8) (3, 7.2) (4, 4.1) (5, 2.3) (6, 1.3)
};

\addplot [color=red, dashed, mark=triangle*, mark size=1.5] coordinates {(1, 18.3) (2, 10.5) (3, 6.1) (4, 3.5) (5, 2.0) (6, 1.1)
};

\addplot [color=black, loosely dotted, mark=none, mark size=1.5] coordinates {(1, 24.8) (2, 14.2) (3, 8.1) (4, 4.6) (5, 2.6) (6, 1.5)
};

\addplot [color=gray, only marks, mark=*, mark size=3] coordinates {(3, 7.2)
};

\legend{Secrecy Rate Loss, Sensing Accuracy Loss, Theoretical Bound, Practical Choice}
\end{axis}
\end{tikzpicture}
\caption{Secrecy-rate and sensing-accuracy loss versus SIM phase-shifter resolution $b$ (bits), including the theoretical loss bound and a practical 3-bit operating point.}
\label{fig:quantization_impact}
\end{figure}
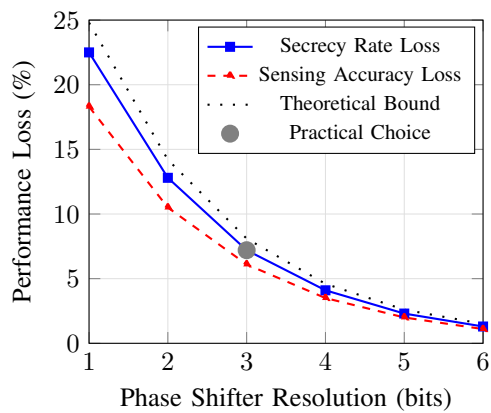


\section{Conclusion}
\label{sec:conclude}
We studied a multi-functional downlink system that integrates secure communication and sensing using a multi-layer SIM, and formulated a multi-objective optimization to trade off maximizing the multiuser sum secrecy rate and minimizing the sensing CRB. To solve the resulting nonconvex problem we developed an enhanced L-BCD algorithm that alternately optimizes hybrid (analog/digital) beamformers, robust artificial-noise covariance, discrete multi-layer SIM phase shifts for sensing/communication/channel estimation, and joint time--power allocation, while handling probabilistic secrecy constraints via Bernstein-type inequality reformulations for robustness to channel uncertainty. Simulations confirm fast convergence, substantial gains over benchmarks (approximately 30--40\%) versus a single-layer RIS with the same element count, and reliable satisfaction of outage-based secrecy constraints, demonstrating a practical, high-performance framework for SIM-assisted integrated sensing and communications.

{\footnotesize
\bibliographystyle{IEEEtran}

}

\end{document}